\documentclass[a4paper]{amsart}
\usepackage{amsthm}
\usepackage{amsmath}
\usepackage{geometry}
\usepackage{amsfonts}
\usepackage{graphicx}
\usepackage{array}
\usepackage{amssymb}

\newtheorem{theorem}{Theorem}
\newtheorem{corollary}[theorem]{Corollary}
\newtheorem{proposition}[theorem]{Proposition}
\newtheorem{definition}[theorem]{Definition}
\newtheorem{lemma}[theorem]{Lemma}\newtheorem{example}[theorem]{Example}

\numberwithin{equation}{section}
\numberwithin{theorem}{section}

\begin{document}
\title[On Low's reconstruction theorem]{\textsc{On the space of light rays of a space-time and a reconstruction theorem by Low}}
\author{A. Bautista, A. Ibort}
\address{Depto. de Matem\'aticas, Univ. Carlos III de Madrid, Avda. de la
Universidad 30, 28911 Legan\'es, Madrid, Spain.}
\email{abautist@math.uc3m.es, albertoi@math.uc3m.es}
\author{J. Lafuente}
\address{Depto. de Geometr\'{\i}a y Topolog\'{\i}a, Univ. Complutense de
Madrid, Avda. Complutense s/n, 28040 Madrid, Spain.}
\email{jlafuente@mat.ucm.es}
\date{}
\thanks{This work has been partially supported by the Spanish MICIN grant
MTM 2010-21186-C02-02 and QUITEMAD P2009 ESP-1594.  A.I. wants to thank the program ``Salvador de Madariaga''
for partial support during the stay at the Dept. of Maths. Univ. California at Berkeley where part of this work was done..}

\begin{abstract}   A reconstruction theorem in terms of the
topology and geometrical structures on the spaces of light rays and skies of a given space--time is discussed.
This result can be seen as part of Penrose and Low's programme intending to describe the causal structure
of a space--time $M$ in terms of the topological and geometrical properties of the space of light rays, i.e., unparametrized time-oriented null geodesics, $\mathcal{N}$.
In the analysis of the reconstruction problem  it becomes instrumental the structure of the space of skies, i.e., of congruences of
light rays.
It will be shown that the space of skies $\Sigma$ of a strongly causal skies distinguishing space--time $M$ carries a
canonical differentiable structure diffeomorphic to the original manifold $M$.
Celestial curves, this is, curves in $\mathcal{N}$ which are everywhere tangent to skies, play a fundamental role
in the analysis of the geometry of the space of light rays.  It will be shown that a celestial curve is induced by a
past causal curve of events iff the legendrian isotopy defined by it is non-negative.  This result extends in a nontrivial way some recent results by Chernov \emph{et al} on Low's Legendrian conjecture.   Finally, it will be shown
that a celestial causal map between the space of light rays of two
strongly causal spaces (provided that the target space is null non--conjugate) is necessarily induced
from a conformal immersion and conversely.  These results make explicit the fundamental role played by the collection of skies, a collection of legendrian spheres
with respect to the canonical contact structure on $\mathcal{N}$, in characterizing the causal structure of space--times.
\end{abstract}

\maketitle
\tableofcontents

\section{Introduction}

In this paper the problem of reconstructing a space-time $M$ from
the topology and geometry of its space of future oriented, unparametrized null geodesics $\mathcal{N}$
or, for brevity, light rays, will be addressed.  This problem can be seen as part of a programme proposed by R. Penrose and
developed partially by R. Low in which a systematic discussion of causality
properties of Lorentzian space--times in terms of the topology of the
corresponding spaces of null geodesics \cite{Lo88}, \cite{Lo90}, \cite{Lo94}, \cite{Lo06} is intended.   Low's conjecture that states that two events
in a time--oriented Lorentzian manifold are causally related iff their corresponding skies, which are legendrian knots with respect to the canonical contact structure in the space of null geodesics, are
linked, constitutes one of its most salient outcomes. Recently it was shown by Chernov and Rudyak \cite%
{Ch08} and Chernov and Nemirovski \cite{Ch10} that Low's conjecture is
actually true in a globally hyperbolic space with a Cauchy surface whose
universal covering is diffeomorphic to an open domain in $\mathbb{R}^{n}$.
Thus the exploration of the relation between the causal properties of a
conformal class of Lorentzian metrics and the topological properties
of skies in the manifold of light rays opens a new and
exciting relation between the topology and causality relations of Lorentzian
space--times and the topology of contact manifolds.

In this paper we will analyze a theorem sketched in Low's papers on the
possibility of recovering the conformal structure of the original
space--time from the space of skies which constitutes a family of Legendrian (possibly
linked) spheres in the contact manifold of  light rays of
the original manifold. Such theorem provides a way to \textquotedblleft come
back\textquotedblright\ from the space of  light rays to the conformal
structure that could contribute to clarify the relation between causality and topological linking.

In the analysis presented here a paramount role is played by the space of skies $\Sigma$
 of the space--time $M$ where the sky $S(x)$ of a given point $x\in M$
is the congruence of  light rays passing through it.
It is well--known that if the space-time $M$, i.e., a time--oriented Lorentzian manifold, is strongly causal then the space of light rays has a smooth structure \cite{Lo89}.  Moreover if we assume that the space $M$ is sky distinguishing, this
is $S(x)\neq S(y)$ if $x\neq y$, then it will be shown  (Section 3, Thm. \ref{theorem2} and Cor. \ref{corollary3.12}) that the space of skies $\Sigma$ carries
a canonical topology as well as a canonical differentiable structure
defined using exclusively the contact structure of the manifold $\mathcal{N}$ and
that such smooth structure is diffeomorphic to the smooth structure of
the original space--time (Corollary 3.9).  The proof of these results are based on
the construction of a basis for the topology of the space of skies by 
regular open subsets of $\Sigma$ where regular means that the corresponding
tangent spaces to the skies elements of the open set ``pile up'' nicely
defining a regular submanifold on the tangent space to $M$.   The proof of this
statement constitutes the main part of section 3, Thm. 3.6, where a new technique
of convergence of families of Jacobi fields is used.

Now we will turn our strategy to study under what circumstances a smooth map
between the spaces of light rays corresponding to two space-times induces
a conformal transformation among them or, in other words, we would like to explore in what sense
the space of light rays of a given space--time characterizes it.   
It is clear that such a map should satisfy strong conditions. 
We will show that such analysis relies heavily on the study of celestial curves.
A celestial curve is a regular curve in $\mathcal{N}$ whose velocity
vector is always tangent to some sky. These curves induce legendrian
isotopies between skies.  It will be shown in Sections 5 and 6, Thms. \ref{cor00300} and \ref{causal_legendrian},
that a curve $\Gamma$ in $\mathcal{N}$ is a causal celestial curve iff
it defines a non--negative legendrian isotopy of skies.  This result extends in
a  non-trivial way results obtained by Chernov \emph{et al} in their analysis
of Low's legendrian conjecture \cite{Ch10}.

Finally the uniqueness of the reconstruction will be discussed In Section 6.
It is clear that diffeomorphisms on $\mathcal{N}$ preserving skies, i.e., inducing  
a diffeomorphism in the original space--time, obviously
preserve celestial curves.    
Then it will be shown that, if we have two strongly causal space--times $M_1$
and $M_2$ such that their spaces of  light rays are diffeomorphic by a
diffeomorphism that transforms causal celestial curves into causal celestial curves, then it induces a conformal immersion $M_1\subset M_2$ provided that the space $M_2$ is null non-conjugate, this is there are no
conjugate points along null geodesic segments.
This theorem provides the uniqueness result we were looking for and it is the best 
that can be obtained as the discussion of the example at the end of this section shows.

\section{The space of  light rays of a space--time: its differentiable and contact structure}

Throughout this section, following the flavour of \cite{Lo89}, \cite{Lo00}
and \cite{Lo06}, we will describe the space of  light rays of a
space--time, its contact structure and atlas for its tangent
bundle that will be useful in what follows.

\subsection{The smooth structure of the space of  light rays}

Let us consider a time--oriented $m$-dimensional Lorentz manifold $M$ with
metric $\mathbf{g}$ and conformal metric class $\mathcal{C}$ (we will just call $(M,\mathcal{C})$ a space--time in what
follows). We will denote, as indicated in the introduction, by $\mathcal{N}$ the space of future oriented
unparametrized null geodesics, or simply light rays, in $M$. We are interested in the causal
structure $\mathcal{C}$ and the selected metric $\mathbf{g}\in\mathcal{C}$
should be considered as an auxiliary tool to study $\mathcal{C}$.

Let us denote by $TM$ the tangent bundle of $M$ and by $\pi
_{M}:TM\rightarrow M$ the corresponding canonical projection. The set $%
\mathbb{N}^{+}=\{\xi \in TM:\mathbf{g}\left( \xi ,\xi \right) =0,\xi \neq 0,\xi
\,\,\mathrm{future}\}\subset TM$ defines the subbundle of future null
vectors over $M$. Any element $\xi \in \mathbb{N}^{+}$ defines a unique future
oriented null geodesic $\gamma $ in $M$ such that $\gamma \left( 0\right)
=\pi _{M}\left( \xi \right) $ and $\gamma ^{\prime }\left( 0\right) =\xi $.
Consider the quotient space of $\mathbb{N}^{+}$ with respect to positive scale
transformations, i.e., the quotient space with respect to the dilation, or
Euler vector field $\Delta $ on $\mathbb{N}^{+}$, that is the space of leaves of
the vector field whose flow is given by $e^{t}\xi $, $t\in \mathbb{R}$. In
this way, we obtain the bundle $\mathbb{PN}^{+}$ of future null directions
\begin{equation*}
\mathbb{PN}^{+}=\{\left[ \xi \right] :u\in \left[ \xi \right] \Leftrightarrow
u=\lambda \xi \text{ where }0\neq \lambda \in \mathbb{R}^{+},\xi \in \mathbb{%
N}\}
\end{equation*}%
Now, any $\left[ \xi \right] \in \mathbb{PN}^{+}$ defines an unparametrized
future oriented null geodesic, i.e., a light ray, in $M$ which is the image in $M$ of the null
geodesic $\gamma $ defined by $\xi \in \mathbb{N}^{+}$. We denote by $\pi \colon
\mathbb{PN}^{+}\rightarrow M$ the canonical projection of the bundle $\mathbb{PN}^{+}
$ over $M$. The fibre $\pi ^{-1}(p)$ is diffeomorphic to the standard
sphere $\mathbb{S}^{m-2}$. We observe that the bundle $\mathbb{PN}^{+}$ is
foliated by the lifts of these light rays to $\mathbb{PN%
}^{+}$, which are projections to $\mathbb{PN}^{+}$ of integral curves of the
geodesic spray $X_{\mathbf{g}}$ restricted to $\mathbb{N}^{+}$. We will call $%
\mathcal{F}$ to this foliation. Then, the space of  light rays $\mathcal{N%
}$ can be defined too as the quotient space $\mathbb{PN}^{+}/\mathcal{F}$ or,
equivalently, as the quotient space of $\mathbb{N}^{+}$ by the foliation $%
\mathcal{K}$ whose leaves are the maximal integral submanifolds lying in $%
\mathbb{N}^{+}$ of the integrable distribution defined by $\Delta $ and $X_{%
\mathbf{g}}$, this is $\mathcal{N}\cong \mathbb{PN}^{+}/\mathcal{F}=\mathbb{N}^{+}/%
\mathcal{K}$. We will denote by $\sigma $ the canonical projection $\sigma
\colon \mathbb{PN}^{+}\rightarrow \mathbb{PN}^{+}/\mathcal{F}$ .

The quotient space $\mathbb{PN}^{+}/\mathcal{F}$ is not a differentiable
manifold in general. It is not hard to construct examples (see for instance
examples 2.1 and 2.2 in \cite{Lo89}) of spaces of  light rays whose
topology cannot be induced by any differentiable structure or which are
non-Hausdorff.  Sufficient conditions are given in \cite[%
Proposition 2.1 and 2.2]{Lo89} that guarantee that $\mathcal{N}$ inherits a
differentiable structure.

\begin{proposition}
Let $M$ be a strongly causal space-time of dimension $m$. Then $\mathbb{PN}^{+}/%
\mathcal{F}$ inherits a canonical differentiable structure from $\mathbb{PN}^{+}$
of dimension $2m-3$ such that $\sigma$ is a smooth submersion. Moreover, if $%
M$ is not nakedly singular, then $\mathbb{PN}^{+}/\mathcal{F}$ is Hausdorff.
\end{proposition}

Hence, for any strongly causal space--time $M$ without naked singularities,
the space of light rays $\mathcal{N}$ inherits the
structure of a Hausdorff smooth $\left(2m-3\right)$-dimensional differentiable
manifold via the natural identification of $\mathcal{N}$ with $\mathbb{PN}^{+}/%
\mathcal{F}$ and $\sigma:\mathbb{PN}^{+}\rightarrow \mathcal{N}$ is a
submersion. Thus in what follows we will assume that $M$ is a strongly
causal not nakedly singular space--time and we call the space of
light rays $\mathcal{N}$ equipped with
the smooth structure above, the space of  light rays of $M$ (see also for instance
\cite{Po12} for a recent discussion on the topology of the space
of causal curves and separation axioms).

Given a point $p\in M$, the set of light rays passing through $p$ will be called \emph{the sky of} $p$ and it
will be denoted by $S\left( p\right) $, i.e.
\begin{equation*}
S\left( p\right) =\{\gamma \in \mathcal{N}:p\in \gamma \subset M\}.
\end{equation*}%
Notice that the geodesics $\gamma \in S(p)$ are in one--to--one
correspondence with the elements in the fiber $\pi ^{-1}(p)\subset \mathbb{PN%
}^{+}$, hence the sky $S\left( p\right) $ of any point $p\in M$ is diffeomorphic
to the standard sphere $\mathbb{S}^{m-2}$. Now, it is possible to define the
\emph{space of skies} as
\begin{equation*}
\Sigma =\{X\subset \mathcal{N}:X=S\left( p\right) \text{ for some }p\in M\}
\end{equation*}%
and the \emph{sky map} as the application $S:M\rightarrow \Sigma $ that maps every $p$ to $%
S\left( p\right) \in \Sigma $. This sky map $S$ is, by definition of $\Sigma
$, surjective.  If the sky map $S$ is a bijection, its inverse map denoted by
$P=S^{-1}:\Sigma \rightarrow M$ will be called the \emph{parachute map}. An
important part of this paper will be devoted to the study of the natural topological
and differentiable structures on the sky space $\Sigma $ considered as a
collection of subsets of $\mathcal{N}$. In order to understand better the
structures inherited by $\Sigma $ we need to analyze the structure of $T%
\mathcal{N}$ and in particular the canonical contact distribution carried by
it.

\subsection{The tangent bundle and the contact structure on the space of  light rays}

Let us consider $\gamma \in \mathcal{N}$, a tangent vector to $\mathcal{N}$ at $%
\gamma $ is defined by an equivalence class $\Gamma ^{\prime }(0)$ of smooth
curves $\Gamma (s)=\gamma _{s}\in \mathcal{N}$, $s\in (-\epsilon ,\epsilon )$
such that $\Gamma (0)=\gamma $. Choosing a auxiliary metric $\mathbf{g}$
in $\mathcal{C}$, we consider the space $\mathcal{J}(\gamma )$ of
Jacobi fields $J(t)$ along the parametrized geodesics $\gamma (t)$, i.e.,
vector fields along the curve $\gamma (t)$ which are tangent to geodesic
variations $\Gamma (s,t)=\gamma _{s}(t)$ of $\gamma (t)$, $J(t)=\partial
\gamma _{s}(t)/\partial s\mid _{s=0}$, then there is a canonical projection $%
\pi _{\gamma }\colon \mathcal{J}(\gamma )\rightarrow T_{\gamma }\mathcal{N}$
given by $\pi _{\gamma }(J)=\Gamma ^{\prime }(0)$, however such map has a
two--dimensional kernel defined by the Jacobi fields of the form $%
(at+b)\gamma ^{\prime }(t)$. If we denote such Jacobi fields by $\mathcal{J}%
_{\mathrm{tan}}(\gamma )$, then a tangent vector to $\mathcal{N}$ at $\gamma
$ can be identified with an equivalence class $[J]=J+\mathcal{J}_{\mathrm{tan%
}}(\gamma )$, with $J\in \mathcal{J}(\gamma )$. Notice that a vector field $%
J $ along the curve $\gamma (t)$ is a Jacobi field if and only if it
satisfies the Jacobi equation:
\begin{equation}
J^{\prime \prime }+R(J,\gamma ^{\prime })\gamma ^{\prime }=0  \label{jacobi}
\end{equation}%
where \textquotedblleft prime\textquotedblright\ in $J$ means the covariant
derivative with respect the Levi--Civita connection defined by $\mathbf{g}$
along the curve $\gamma (t)$.
Then it follows immediately that any Jacobi vector field
$J(t)$ defined by a geodesic variation $\gamma _{s}(t)$ in $\mathcal{N}$
satisfies
\begin{equation}\label{orthogonal}
\mathbf{g}\left( J(t),\gamma ^{\prime }(t)\right) = \mathrm{constant}.
\end{equation}
In what follows we will identify a Jacobi field $J(t)$ along $\gamma (t)$
with a tangent vector at $\gamma $ understanding by it the equivalence class
$[J]$, i..e, $J(\mathrm{mod}\gamma ^{\prime })$.

There exists a contact structure in $\mathcal{N}$ which arises from the
canonical 1-form $\theta$ on $T^{*}M$ but that can be described explicitly
in terms of Jacobi fields \cite{Lo98}, \cite{Lo06}.
Define for each $\gamma\in \mathcal{N}$ the hyperplane $\mathcal{%
H}_\gamma \subset T_\gamma\mathcal{N}$ given by:
\begin{equation}  \label{contact}
\mathcal{H}_{\gamma}=\lbrace J\in T_{\gamma}\mathcal{N}:\mathbf{g}%
\left(J,\gamma ^{\prime }\right)=0 \rbrace .
\end{equation}

\begin{proposition}
The distribution $\mathcal{H} = \bigcup_{\gamma \in \mathcal{N}} \mathcal{H}%
_{\gamma}$ defines a contact structure on $\mathcal{N}$.
\end{proposition}

The proof of the previous proposition takes advantage of the fact that $%
\mathcal{N}$ has been constructed from $TM$, but it is more
convenient to start from $T^{*}M$ via the diffeomorphism defined by the
metric $\mathbf{g}$. Hence, if $\hat{\mathbf{g}}\colon TM \to T^*M$ denotes
the canonical diffeomorphism defined by the metric $\mathbf{g}$, then $\hat{%
\mathbf{g}}(X_{\mathbf{g}}) = X_H$ is just the Hamiltonian vector field
corresponding to the kinetic energy Hamiltonian $H (x,p)$ on $T^*M$ and $%
\hat{\mathbf{g}}(\Delta)$ is just the Euler field on $T^*M$. But $T^*M$
carries a canonical 1-form $\theta$, its Liouville 1--form. Then we may
restrict $\theta$ to $\mathbb{N}^{+*} := \hat{\mathbf{g}}(\mathbb{N}^{+})$, whose
kernel defines a field $\ker \theta$ of hyperplanes on $\mathbb{N}^{+*}$. The
distribution $\ker \theta$ is invariant with respect to the flow of the Euler vector field
$\Delta$ on $T^*M$ because $L_\Delta \theta = \theta$ and it is also
invariant under the flow of $X_H$ because $L_{X_H}\theta = 0$, so $\ker
\theta$ descends to $\mathbb{PN}^{+*}$ and then to $\mathcal{N}$. This
defines the contact structure \eqref{contact} on $\mathcal{N}$.

Actually, if we denote by $\tilde{\sigma}$ the canonical
projection $\tilde{\sigma} \colon \mathbb{N}^{+*} \to \mathcal{N}$, $\tilde{%
\sigma}(x,p) = \gamma$ where $\gamma$ is the projection on $M$ of the
integral curve of $X_H$ passing at time 0 through $(x,p)$, i.e., $\gamma$ is
the geodesic such that $\gamma(0) = x$ and $\gamma^{\prime }(0) = v$ with $%
\hat{\mathbf{g}}(v) = p$, then a tangent vector $(\dot{x}, \dot{p}) \in
T_{(x,p)}\mathbb{N}^{+*}$ will be in the $\ker\theta$ iff $\langle p,\dot{x}
\rangle = 0$. The tangent vector $(\dot{x}, \dot{p})$ is mapped by $\tilde{%
\sigma}$ into a tangent vector $J$ to $\mathcal{N}$, hence we get eq. %
\eqref{contact}.

Moreover, if $\gamma \in X=S\left( p\right) $ where $X$ is the sky of $p\in
M $ with $\gamma \left( s_{0}\right) =p$, then
\begin{equation}
T_{\gamma }X=\{J\in T_{\gamma }\mathcal{N}:J\left( s_{0}\right) =0\left(
\mathrm{{mod}\gamma ^{\prime }}\right) \}.  \label{tangent_sky}
\end{equation}%
For any $J\in T_{\gamma }X$, since $\mathbf{g}\left( J,\gamma ^{\prime
}\right) $ must be constant and $J\left( s_{0}\right) =0\left( \mathrm{{mod}%
\gamma ^{\prime }}\right) $, then $\mathbf{g}\left( J,\gamma ^{\prime
}\right) =0$ and therefore $T_{\gamma }X\subset \mathcal{H}_{\gamma }$. This
implies that any $T_{\gamma }X$ is a subspace of $\mathcal{H}_{\gamma }$ and
moreover because $\dim X = m -2$, $X$ is a Legendrian manifold of the contact structure on $\mathcal{N%
}$.

\subsection{A smooth atlas for the tangent bundle of the space of  light rays}\label{charts}

We will construct now an atlas for the tangent bundle $T\mathcal{N}$ that is well adapted
to the causal structure of $M$ in the sense that in its definition we will take
advantage that given an event $p$ in a strongly causal space--time $M$ we
can always choose a globally hyperbolic causally convex normal neighborhood $%
V$ of $p$ (see for instance \cite[Thm. 2.1 and Def. 3.22]{Mi08}). Notice that being $V$
causally convex then for any null geodesic $\gamma$ we have that $\gamma\cap
V$ is connected.

First we will consider an atlas for $M$ whose local charts are $\left(
V,\varphi =\left( \mathbf{x}^{1},\ldots \mathbf{x}^{m}\right) \right) $ with
$V$ a globally hyperbolic causally convex normal neighbourhood such that,
without lack of generality, the local hypersurface $C\subset V$ defined by $%
\mathbf{x}^{1}=0$ is a smooth spacelike (local) Cauchy surface, hence each
null geodesic cutting $V$ intersects $C$ at exactly one point. Let $\left\{
E_{1},\ldots ,E_{m}\right\} $ be an orthonormal frame in $V$ such that $%
E_{1} $ is a future oriented timelike vector field in $V$. If $\xi \in T_{p}V$ is written as
$\xi =\sum\limits_{j=1}^{m}\mathbf{u}^{j}E_{j}\left( p\right) $ then $%
(TV,\Phi )$ with:
\begin{equation}
\Phi \colon TV\rightarrow \mathbb{R}^{m};\quad \xi \mapsto \left( \mathbf{x}%
^{1},\ldots ,\mathbf{x}^{m},\mathbf{u}^{1},\ldots ,\mathbf{u}^{m}\right)
\label{equation3.1}
\end{equation}%
is a local coordinate chart in $TM$. Let us denote by $\mathbb{N}^{+}\left(
V\right)$ the
restriction of the bundle $\mathbb{N}^{+}$ to $V$ and by $\mathbb{PN}^{+}\left(
V\right) =\{\left[ \xi \right] \in \mathbb{PN}^{+}:\pi _{M}\left( \left[ \xi %
\right] \right) \in V\}$ the same for $\mathbb{PN}^{+}$. For $\xi \in
\mathbb{N}^{+}\left( V\right) $ we have $\left( \mathbf{u}^{1}\right)
^{2}=\sum\limits_{j=2}^{m}\left( \mathbf{u}^{j}\right) ^{2}$ so, a
coordinate chart in $\mathbb{N}^{+}\left( V\right) $ is given by the map
\begin{equation}
\xi \mapsto \left( \mathbf{x}^{1},\ldots ,\mathbf{x}^{m},\mathbf{u}%
^{2},\ldots ,\mathbf{u}^{m}\right) \in \mathbb{R}^{2m-1}  \label{equation3.2}
\end{equation}%
Taking now homogeneous coordinates $\left[ \mathbf{u}%
^{1},\ldots ,\mathbf{u}^{m}\right]$ for $\left[ \xi \right] \in \mathbb{PN}^{+}%
\left( V\right) $ in \eqref{equation3.2}, or equivalently, fixing $\mathbf{u}%
^{1}=1$ then $\left( \mathbf{u}^{2},\ldots ,\mathbf{u}^{m}\right) $ lies in $%
\mathbb{S}^{m-2}$ and describes a null direction. So, in this way, taking
for example $\mathbf{u}^{2}=\sqrt{1-\left( \mathbf{u}^{3}\right) ^{2}-\cdots
\left( \mathbf{u}^{m}\right) ^{2}}$ we obtain the coordinate chart $[\Phi
]\colon \mathbb{PN}^{+}\left(V\right)\rightarrow \mathbb{R}^{2m-2}$ defined as:
\begin{equation}
\left[ \xi \right] \mapsto \left( \mathbf{x}^{1},\ldots ,\mathbf{x}^{m},%
\mathbf{u}^{3},\ldots ,\mathbf{u}^{m}\right) \in \mathbb{R}^{2m-2}
\label{equation3.3}
\end{equation}%
for $\mathbb{PN}^{+}\left( V\right) $. Let $\mathcal{U}$ be the image of the
projection $\sigma :\mathbb{PN}^{+}\left( V\right) \mapsto \mathcal{N}$. Clearly
$\mathcal{U}\subset \mathcal{N}$ is open. By global hyperbolicity of $V$,
every null geodesic passing through $V$ intersects $C$ at a unique point and
this assures that $\sigma \left( \mathbb{PN}^{+}\left( V\right) \right) =\sigma
\left( \mathbb{PN}^{+}\left( C\right) \right) =\mathcal{U}$. We have assumed
that the Cauchy surface $C$ is a smooth regular submanifold of $V$, this
implies that the bundle $\mathbb{PN}^{+}\left( C\right) $ is a smooth regular
submanifold of $\mathbb{PN}^{+}\left( V\right) $, moreover the map $\left.
\sigma \right\vert _{\mathbb{PN}^{+}\left( C\right) }:\mathbb{PN}^{+}\left( C\right)
\mapsto \mathcal{U}$ is a differentiable bijection. The map $\sigma $ is a
submersion such that, for any $\left[ \xi \right] \in \mathbb{PN}^{+}\left(
V\right) $, the kernel of $d\sigma _{\left[ \xi \right] }$, is the
one--dimensional subspace generated by tangent vectors to curves defining
 light rays, i.e. curves $\lambda \left( s\right) =\left[ \gamma ^{\prime
}\left( s\right) \right] \in \mathbb{PN}^{+}_{\gamma \left( s\right) }$ where $%
\gamma $ is a null geodesic and $\left[ \gamma ^{\prime }\left( s\right) %
\right] =\{\lambda \gamma ^{\prime }\left( s\right) :\lambda \in \mathbb{R}%
\} $.   Because $C$ is a spacelike surface, the kernel of $\left. d\sigma _{\left[
\xi \right] }\right\vert _{\mathbb{PN}^{+}\left( C\right) }$ is trivial, hence $%
\left. d\sigma _{\left[ \xi \right] }\right\vert _{\mathbb{PN}^{+}\left(
C\right) }$ is a surjection between vector spaces of the same dimension,
therefore $\left. \sigma \right\vert _{\mathbb{PN}^{+}\left( C\right) }$ is a
diffeomorphism. We have the following diagram
\begin{equation*}
\begin{tabular}{ccc}
$\mathbb{PN}^{+}\left( V\right) $ & $\overset{\sigma }{\longrightarrow }$ & $%
\mathcal{U}$ \\
$inc\uparrow $ & $\nearrow $ &  \\
$\mathbb{PN}^{+}\left( C\right) $ &  &
\end{tabular}%
\end{equation*}%
So, we can use the restriction of the chart \eqref{equation3.3} to $\mathbb{PN}^{+}%
\left( C\right) $ as a coordinate chart in $\mathcal{U}\subset \mathcal{N}$.
This coordinate chart in $\mathcal{U}$ is given by the map $\psi \colon
\mathcal{U}\rightarrow \mathbb{R}^{2m-3}$:
\begin{equation}
\gamma \mapsto \psi (\gamma )=\left( \mathbf{x}^{2},\ldots ,\mathbf{x}^{m},%
\mathbf{u}^{3},\ldots ,\mathbf{u}^{m}\right) =\left( \mathbf{x},\mathbf{u}%
\right) \in \mathbb{R}^{m-1} \times \mathbb{R}^{m-2} = \mathbb{R}^{2m-3} \label{equation3.4}
\end{equation}%
with $\mathbf{x}=(\mathbf{x}^{2},\ldots ,\mathbf{x}^{m})$ and $\mathbf{u}=(%
\mathbf{u}^{3},\ldots ,\mathbf{u}^{m})$, where $\gamma \left( 0\right) =p\in C\subset
V$ have coordinates $\mathbf{x}$ and $\gamma ^{\prime
}\left( 0\right) = \xi = \sum_{i=1}^m \mathbf{u}^{i}E_{i}\in \mathbb{N}^{+}\left(
C\right) $.

We will define an atlas on $T\mathcal{N}$ by using the open sets $T\mathcal{U%
}$ over the open sets $\mathcal{U}$ defined above. Thus, in order to
complete a chart in $T\mathcal{U}$, we will add
the coordinates for the tangent vectors at every null geodesic $\gamma \in
\mathcal{N}$ with coordinates $\mathbf{x},\mathbf{u}$.
This can be done by using the initial values at $\ t=t_{0} = 0$
for Jacobi's equation \eqref{jacobi} whose solutions are the Jacobi fields along $\gamma $%
. Thus if $J\in T_{\gamma }\mathcal{N}$ then $J\left( t_{0}\right)
=\sum\limits_{j=1}^{m}\mathbf{w}^{j}E_{j}\left( p\right) $ and $J^{\prime
}\left( t_{0}\right) =\sum\limits_{j=1}^{m}\mathbf{v}^{j}E_{j}\left(
p\right) $ define $J$, so a chart in $T\mathcal{U}$ is given by the map $%
\overline{\psi }$:
\begin{equation}
J\mapsto \overline{\psi }(J)=\left( \mathbf{x},\mathbf{u}; \left\langle
\mathbf{v}^{1},\ldots ,\mathbf{v}^{m}\right\rangle ,\left\langle \mathbf{w}%
^{1},\ldots ,\mathbf{w}^{m}\right\rangle \right) = \left( \mathbf{x},\mathbf{u%
} ; \mathbf{v},\mathbf{w}\right) \in \mathbb{R}^{4m-6}  \label{equation3.5}
\end{equation}%
with $\mathbf{v} = \left\langle \mathbf{v}^{1},\ldots ,\mathbf{v}%
^{m}\right\rangle := \left( \mathbf{v}^{1},\ldots ,\mathbf{v}^{m}\right)
\left( \mathrm{mod}\gamma ^{\prime }\right) $ and $\mathbf{w} = \left\langle
\mathbf{w}^{1},\ldots ,\mathbf{w}^{m}\right\rangle :=\left( \mathbf{w}%
^{1},\ldots ,\mathbf{w}^{m}\right) \left( \mathrm{mod}\gamma ^{\prime
}\right) $.

The notation $\left( \mathbf{a}^{1},\ldots ,\mathbf{a}%
^{m}\right) \left( \mathrm{mod}\gamma ^{\prime }\right) $ means $%
\sum\limits_{j=1}^{m}\mathbf{a}^{j}E_{j}\left( p\right) $ $\left( \mathrm{mod%
}\gamma ^{\prime }\left( t_{0}\right) \right) $.
We may define $m-2$
independent coordinates from $( \mathbf{v}^{1},\ldots ,\mathbf{v}^{m} )$ and $%
m-1$ from $(\mathbf{w}^{1},\ldots ,\mathbf{w}^{m})$.  Notice that  because of \eqref{orthogonal}, $J^{\prime }\left(t_{0}\right)$ is orthogonal to $\gamma^{\prime
}\left(t_{0}\right)$, so $\mathbf{v}^{1}=\mathbf{v}^{2}\mathbf{u}%
^{2}+\cdots +\mathbf{v}^{m}\mathbf{u}^{m}$.   Then, we may consider the representatives $\overline{J}, \overline{J}' \in T\mathcal{N}$ of $J\left(t_0\right)$ and $J '\left(t_0\right)$ respectively as 
\begin{equation}
\overline{J}=J\left(t_0\right)-\mathbf{w}^{1}\gamma'\left(t_0\right)=\left(\mathbf{w}^{2}-\mathbf{w}^{1}\mathbf{u}^{2}\right)E_2 + \cdots +\left(\mathbf{w}^{m}-\mathbf{w}^{1}\mathbf{u}^{m}\right)E_m
\end{equation} 
\begin{equation}
\overline{J}'=J'\left(t_0\right)-\mathbf{v}^{1}\gamma'\left(t_0\right)=\left(\mathbf{v}^{2}-\mathbf{v}^{1}\mathbf{u}^{2}\right)E_2 + \cdots +\left(\mathbf{v}^{m}-\mathbf{v}^{1}\mathbf{u}^{m}\right)E_m
\end{equation} 
therefore the coordinates $\mathbf{v}$ and $\mathbf{w}$ can be written as 
\begin{equation}
\left\{
\begin{tabular}{c}
$\mathbf{v}= \left(\overline{\mathbf{v}}^{3},\ldots , \overline{\mathbf{v}}^{m} \right)$ \\
$\mathbf{w}= \left(\overline{\mathbf{w}}^{2},\ldots , \overline{\mathbf{w}}^{m} \right)$
\end{tabular}
\right.
\end{equation} 
where $\overline{\mathbf{v}}^{k}=\mathbf{v}^{k}-\mathbf{v}^{1}\mathbf{u}^{k}$ and $\overline{\mathbf{w}}^{k}=\mathbf{w}^{k}-\mathbf{w}^{1}\mathbf{u}^{k}$ for $k=1,\ldots,m$. 
It is fair that if, for instance, $\mathbf{u}^{2}\neq 0$ then $\overline{\mathbf{v}}^{2}=\frac{-1}{\mathbf{u}^{2}}\sum_{j=3}^{m} \overline{\mathbf{v}}^{j}\mathbf{u}^{j}$ since $\mathbf{v}^{1}=\mathbf{v}^{2}\mathbf{u}^{2}+\cdots +\mathbf{v}^{m}\mathbf{u}^{m}$. 
So, we will denote, with
an slight abuse of notation, by $\left( \mathbf{x},\mathbf{u}; \mathbf{v},%
\mathbf{w}\right) $ the $4m-6$ independent coordinates thus constructed on $T%
\mathcal{U}$.

It is possible to show the compatibility between the canonical atlas defined
on the tangent bundle $T\mathcal{N}$ over the open sets $T\mathcal{U}$ with
canonical coordinates $\left(\mathbf{x}, \mathbf{u}, \mathbf{\dot{x}},
\mathbf{\dot{u}}\right)$ and the atlas defined by the local charts $%
\left( \mathbf{x},\mathbf{u},\mathbf{v},\mathbf{w} \right)$ defined previously.  Actually
in doing so we will show that the local charts  $\left( \mathbf{x},\mathbf{u},\mathbf{v},\mathbf{w} \right)$
define an atlas.  We prove first the following simple lemma.

\begin{lemma}
\label{lemmaDC92} Let $M$ be a Lorentz manifold, $\gamma:\left[0,1\right]%
\rightarrow M$ a null geodesic, $\lambda :\left( -\epsilon ,\epsilon
\right)\rightarrow M$ a curve verifying that $\lambda \left( 0\right)
=\gamma\left( 0\right)$, and $W\left( s\right) $ a null vector field along $%
\lambda $ such that $W\left( 0\right) =\gamma ^{\prime }\left( 0\right) $.
Then the family of curves:
\begin{equation*}
\mathbf{f}\left( s,t\right) =\mathrm{exp}_{\lambda \left( s\right)
}\left(tW\left( s\right) \right)
\end{equation*}%
is a geodesic variation of $\gamma (t)$ through  light rays with $\mathbf{%
f}\left( 0,t\right) =\gamma \left( t\right)$ and if $J\left( t\right) =\frac{%
\partial \mathbf{f}}{\partial s}\left(0,t\right) $, then
\begin{equation*}
\frac{DW}{ds}\left( 0\right) =\frac{DJ}{dt}\left( 0\right)
\end{equation*}
\end{lemma}

\begin{proof}
On one hand, $\frac{\partial\mathbf{f}}{\partial s}\left(0,0\right)$ is the
tangent vector to the curve $\mathbf{f}\left(s,0\right)$ at $s=0$, and since $%
\mathbf{f}\left(s,0\right)=\mathrm{exp}_{\lambda\left(s\right)}\left(0\cdot
W\left(s\right)\right)= \mathrm{exp}_{\lambda\left(s\right)}\left(0\right)=%
\lambda\left(s\right)$, then we have
\begin{equation*}
J\left(0\right)=\frac{\partial\mathbf{f}}{\partial s}\left(0,0\right)=\frac{d\lambda}{ds}%
\left(0\right)=\lambda ^{\prime }\left(0\right)
\end{equation*}
On the other hand, $\frac{D}{ds}\frac{\partial\mathbf{f}}{\partial t}%
\left(0,0\right)$ is the covariant derivative of the vector field $\frac{%
\partial\mathbf{f}}{\partial t}\left(s,0\right)=W\left(s\right)$ for $s=0$
along the curve $\mathbf{f}\left(s,0\right)=\lambda\left(s\right)$. Then we
can write
\begin{equation*}
\frac{DJ}{dt}\left(0\right)=\frac{D}{dt}\frac{\partial\mathbf{f}}{\partial s}\left(0,0\right)=\frac{D}{ds%
}\frac{\partial\mathbf{f}}{\partial t}\left(0,0\right)= \frac{DW}{ds}%
\left(0\right)
\end{equation*}
therefore $J$ is the Jacobi field of the geodesic variation $\mathbf{f}$.
\end{proof}

Let us consider the coordinate chart $\left(\psi, \mathcal{U}\right)$ in $%
\mathcal{N}$ given by \eqref{equation3.4} where $\gamma\left(0\right)\in C$ for each $\gamma \in \mathcal{U}$.
Now let $\Gamma_1 (s) \in \mathcal{U} \subset \mathcal{N}$, $s \in
(-\epsilon, \epsilon)$, be a curve such that its coordinates are
\begin{equation*}
\psi\left(\Gamma_1\left(s\right)\right) = \left(x_{0}^{2},\dots
,x_{0}^{m},\alpha^{3}\left(s\right),\dots,\alpha^{m}\left(s\right)\right)
\end{equation*}
This curve corresponds to a geodesic variation $\mathbf{f}\left(s,t\right)$
such that
\begin{equation*}
\lambda\left(s\right)=\mathbf{f}\left(s,0\right)=p\in M
\end{equation*}
for every $s$ because the coordinates $\mathbf{x}^k=x_{0}^{k}$ remain
constant. Moreover $\beta\left(s\right)=\frac{\partial\mathbf{f}%
\left(s,t\right)}{\partial t}\in T_{p}M$ is the curve given by
\begin{equation*}
\beta\left(s\right)=E_1\left(p\right)+\alpha^{2}\left(s\right)E_{2}\left(p%
\right)+ \alpha^{3}\left(s\right)E_{3}\left(p\right)+\ldots +
\alpha^{m}\left(s\right)E_{m}\left(p\right).
\end{equation*}
Hence $\mathbf{f}$ can be written by the expression similar to the one in
Lemma \ref{lemmaDC92}
\begin{equation*}
\mathbf{f}\left( s,t\right) =\mathrm{exp}_{p}\left( t\beta \left(
s\right)\right) \, .
\end{equation*}%
Calling $J$ the Jacobi field of $\mathbf{f}$, then by Lemma \ref{lemmaDC92}
we have that
\begin{equation}
\left\{
\begin{matrix}
J\left( 0\right) =0 \\
J^{\prime }\left( 0\right) =\beta ^{\prime }\left( 0\right)%
\end{matrix}%
\right.  \label{res1}
\end{equation}

Now, if we consider a curve $\Gamma _{2}\subset \mathcal{N}$ such that its
coordinates are
\begin{equation*}
\psi \left( \Gamma _{2}\left( s\right) \right) =\left( x^{2}\left(
s\right) ,\dots ,x^{m}\left( s\right) ,u_{0}^{3},\dots ,u_{0}^{m}\right)
\end{equation*}
This curve corresponds to a geodesic variation $\mathbf{f}\left( s,t\right) $
verifying
\begin{equation*}
\lambda \left( s\right) =\mathbf{f}\left( s,0\right) \in C\subset M
\end{equation*}
The fact of the coordinates $\mathbf{u}^{k}=u_{0}^{k}$ remain constant
implies that
\begin{equation}
W\left( s\right) =\frac{\partial \mathbf{f}}{\partial t}\left( s,0\right)=
E_{1}\left( \lambda \left( s\right) \right) +u_{0}^{2}E_{2}\left(
\lambda\left( s\right) \right) +\ldots +u_{0}^{m}E_{m}\left( \lambda
\left(s\right) \right) \in T_{\lambda \left( s\right) }M  \label{campoW}
\end{equation}
So the geodesic variation $\mathbf{f}$ corresponding to $\Gamma_2$ can be written by
\begin{equation*}
\mathbf{f}\left( s,t\right) =\mathrm{exp}_{\lambda \left( s\right)
}\left(tW\left( s\right) \right)
\end{equation*}%
Again, if $J$ is the Jacobi field of $\mathbf{f}$, then by Lemma \ref%
{lemmaDC92}
\begin{equation}
\left\{
\begin{matrix}
J\left( 0\right) =\lambda ^{\prime }\left( 0\right) \\
J^{\prime }\left( 0\right) =\frac{DW}{ds}\left( 0\right)  \, .%
\end{matrix}%
\right.  \label{res2}
\end{equation}

If we choose the curves $\Gamma_1$ and $\Gamma_2$ such that $%
\Gamma_{1}^{\prime}\left(0\right)=\left(\frac{\partial}{\partial u^i}%
\right)_{\Gamma_{1}\left(0\right)}$ y $\Gamma_{2}^{\prime}\left(0\right)=%
\left(\frac{\partial}{\partial x^j}\right)_{\Gamma_{2}\left(0\right)}$
respectively with $i=3,\ldots,m$ and $j=2,\ldots,m$, then we have that the change from canonical coordinates $\left(%
\mathbf{x},\mathbf{u},\mathbf{\dot{x}},\mathbf{\dot{u}}\right)$ to the
coordinates $\left(\mathbf{x},\mathbf{u},\mathbf{v}, \mathbf{w}\right)$
verifies:

\begin{equation}  \label{matrizcambio}
\left(%
\begin{matrix}
\mathbf{v} \\
\mathbf{w}%
\end{matrix}%
\right)=
\left(%
\begin{matrix}
\overline{\mathbf{v}}^{i} \\
\overline{\mathbf{w}}^{j}%
\end{matrix}%
\right)= \left(%
\begin{matrix}
B & I_{m-2} \\
A & 0%
\end{matrix}%
\right) \left(%
\begin{matrix}
\mathbf{\dot{x}} \\
\mathbf{\dot{u}}%
\end{matrix}%
\right) \text{ with } i=3,\ldots,m \text{ and } j=2,\ldots,m
\end{equation}
and where $I_{m-2}\in\mathbb{R%
}^{\left(m-2\right)\times \left(m-2\right)}$ is the identity matrix and $B\in \mathbb{R%
}^{\left(m-2\right)\times \left(m-1\right)}$ is the matrix whose $%
\left(k-1\right)$-th column is the vector containing the $\mathbf{v}$%
--coordinates of $\frac{D W_k}{ds}\left(0\right)$ with $k=2,\dots,m$ being
\begin{equation}  \label{campoWsub}
W_k\left( s\right) =E_{1}\left( \lambda_k \left( s\right) \right) +
u_{0}^{2}E_{2}\left( \lambda_k \left( s\right) \right)+\ldots
+u_{0}^{m}E_{m}\left( \lambda_k \left(s\right) \right) \in T_{\lambda_k
\left( s\right) }M
\end{equation}
with $\lambda_k\left(s\right)$ a curve such that $\mathbf{x}%
^j\left(\lambda_k\left(s\right)\right)=x_{0}^{j}$ are constant for $j\neq k$
and $\mathbf{x}^k\left(\lambda_k\left(s\right)\right)=x_{0}^{k}+s$.

Since $J\left(0\right)=\lambda_{k}'\left(0\right)=\left(\frac{\partial}{\partial x_k}\right)_{\Gamma_2\left(0\right)}=\sum_{j=1}^{m}\mathbf{w}^{j}_{k}E_j$ then we have that $\overline{\mathbf{w}}^j = \mathbf{w}^{j}_{k}-\mathbf{w}^{1}_{k}\mathbf{u}^j$ for $j=2,\ldots,m$. 
This implies that the matrix $A$ is given by
\begin{equation}
A=\left(\mathbf{w}^{j}_{k}-\mathbf{w}^{1}_{k}\mathbf{u}^j\right) \text{ for } j,k=2,\ldots,m
\end{equation}
Calling $\mathbb{V}=\text{span}\left\{E_j \left(\lambda_{k}\left(0\right)\right)\right\}_{j=2,\ldots,m}$, observe the projection $\pi_{\mathbf{u}}:T_{\lambda_{k}\left(0\right)}M\rightarrow \mathbb{V}$ given by 
\[
\pi_{\mathbf{u}}\left(\eta\right)=\eta -\mathbf{g}\left(\eta,E_1\right)\gamma'\left(0\right)
\]
where we have taken $\gamma'\left(0\right)=E_1+\mathbf{u}^{2}E_2+\cdots +\mathbf{u}^{m}E_m$.
The matrix $\widetilde{A}$ of $\pi_{\mathbf{u}}$ relative to the basis $\left\{\left(\frac{\partial}{\partial x_k}\right)_{\Gamma_2\left(0\right)}\right\}_{k=1,\ldots,m}$ in $T_{\lambda_{k}\left(0\right)}M$ and $\left\{E_j \left(\lambda_{k}\left(0\right)\right)\right\}_{j=2,\ldots,m}$ in $\mathbb{V}$ is 
\[
\widetilde{A} = \left( \mathbf{w}^{j}_{k}-\mathbf{w}^{1}_{k}\mathbf{u}^j\right) \text{ for } j=2,\ldots,m \text{ and } k=1,\ldots,m
\]
We have that $\mathbb{V}$ and $\mathbb{V}_2=\text{span}\left\{\left(\frac{\partial}{\partial x_k}\right)_{\Gamma_2\left(0\right)}\right\}_{k=2,\ldots,m}$ are spacelike by construction, $\text{ker}\pi_{\mathbf{u}}=\text{span}\left\{\gamma'\left(0\right)\right\}$ and the matrix of the restriction $\left. \pi_{\mathbf{u}}\right|_{\mathbb{V}_2}$ is $A$, then $\left. \pi_{\mathbf{u}}\right|_{\mathbb{V}_2}$ is an isomorphism and therefore $A$ is regular. Hence, the matrix in \eqref{matrizcambio} describing the change of coordinates along the fibers
of the tangent bundle $T\mathcal{N}$  is regular and differentiable, then the
change of coordinates
\begin{equation*}
\left( \mathbf{x},\mathbf{u},\mathbf{\dot{x}},\mathbf{\dot{u}}\right)
\longleftrightarrow \left( \mathbf{x},\mathbf{u},\mathbf{v},\mathbf{w}\right)
\end{equation*}%
is also differentiable. This also shows that $\left( \mathbf{x},\mathbf{u},%
\mathbf{v},\mathbf{w}\right) $ is a coordinate chart of the canonical
differentiable structure of $T\mathcal{N}$.

\section{The space of skies: its topology and differentiable structure}\label{differentiable_structure}

Henceforth all the strongly causal manifolds $\left( M,\mathcal{C}\right) $ that we will consider verify, in addition, the property that \emph{skies distinguish points}, i.e., if $x\neq y$ are two different events, then
$S(x)\neq S(y)$ or, in other words, the sky map $S\colon M\rightarrow \Sigma $
is injective, hence a bijection.  Notice that this property is weaker than the non--refocusing property introduced by
Low in \cite{Lo06}.

We will start by defining a natural topology on the space of skies $\Sigma $ induced by the topology of $\mathcal{N}$.

\begin{lemma}
\label{lemma3.1} The collection of sets $\mathfrak{T}=\lbrace U \subset
\Sigma : \mathcal{U}=\bigcup\limits_{X\in U}X$ is open in $\mathcal{N}
\rbrace$ is a topology in $\Sigma $.
\end{lemma}

\begin{proof}
Obviously we have that $\emptyset $ and $\Sigma $ are in $%
\mathfrak{T}$.   If $U_{\alpha }\in \mathfrak{T}$ for every $\alpha
\in I$ then $\bigcup\limits_{\alpha \in I}U_{\alpha }$ is in $%
\mathfrak{T}$. Finally, if $U_{k}\in \mathfrak{T} $ for every $k=1,\ldots ,N$ then
$$ \mathcal{V} = \bigcap\limits_{k=1}^{N} \mathcal{U}_k =  \bigcap\limits_{k=1}^{N}\bigcup\limits_{X\in
U_{k}}X=\bigcup\limits_{X\in \bigcap\limits_{k=1}^{N}U_{k}}X$$
is open in $%
\mathcal{N}$, therefore $V= \bigcap\limits_{k=1}^{N}U_{k}$ is in $\mathfrak{T}
$.
\end{proof}

\begin{definition}
The topology $\mathfrak{T}$ in $\Sigma$ defined in Lemma \ref{lemma3.1} will be
called the \emph{reconstructive topology} of $\Sigma $.
\end{definition}

\begin{lemma}
\label{lemma3.3} Given the reconstructive topology in $\Sigma $, then the
sky map $S:M\rightarrow \Sigma $ is an homeomorphism.
\end{lemma}

\begin{proof}
First, we will show that $S$ is an open map. Let $V\subset M$ be an open set
in $M$ and let $S\left( V\right) =\left\{ S\left( x\right) \in \Sigma :x\in
V\right\} $ be its image through $S$. We have to prove that $\mathcal{U}%
=\bigcup\limits_{x\in V}S\left( x\right) $ is open in $\mathcal{N}$, but
this is equivalent to prove that $\sigma ^{-1}\left( \mathcal{U}\right) $ is
open in $\mathbb{PN}^{+}$ where $\sigma $ is the quotient map $\sigma :\mathbb{PN%
}^{+}\rightarrow \mathcal{N}$ and the set $\mathcal{U\subset N}$
is open if and only if $\sigma ^{-1}\left( \mathcal{U}\right) \subset
\mathbb{PN}^{+}$ is open, but we observe that $\sigma ^{-1}\left( \mathcal{U}\right)
=\pi _{M}^{-1}\left( V\right) $ where $\pi _{M}:\mathbb{PN}^{+}\rightarrow M$
is the canonical projection.

Next, we will show that $S$ is continuous. Consider $U\subset \Sigma $ an
open set, then by bijectivity of $S$, we can write $U=\left\{
S\left( x\right) \in \Sigma :x\in V\right\} $, hence $\mathcal{U}%
=\bigcup\limits_{x\in V}S\left( x\right) $ is open in $\mathcal{N}$. Since
$\pi _{M}$ is open map, proving that that $V=\pi _{M}\left( \sigma ^{-1}\left( \mathcal{U}\right) \right) $
is sufficient to proof that $V$ is open.
Indeed, if $y\in V$ then $S\left( y\right) \in U$ and $S\left(
y\right) \subset \mathcal{U}$, so $\sigma ^{-1}\left( S\left( y\right)
\right) \subset \sigma ^{-1}\left( \mathcal{U}\right) $. Since $\sigma
^{-1}\left( S\left( y\right) \right) $ coincides with the fibre $\mathbb{PN}^{+}%
_{y}$ then $y=\pi _{M}\left( \mathbb{PN}^{+}_{y}\right) =\pi _{M}\left( \sigma
^{-1}\left( S\left( y\right) \right) \right) \in \pi _{M}\left( \sigma
^{-1}\left( \mathcal{U}\right) \right) $. On the other hand, if $y\in \pi
_{M}\left( \sigma ^{-1}\left( \mathcal{U}\right) \right) $ then there exists
$v_{y}\in \mathbb{PN}^{+}_{y}$ such that $v_{y}\in \sigma ^{-1}\left( \mathcal{U}%
\right) $. By the definition of $\mathcal{U}$ as a union of skies, then the
whole fibre $\mathbb{PN}^{+}_{y}$ must be contained in $\sigma ^{-1}\left(
\mathcal{U}\right) $, and hence $S\left( y\right) \subset \mathcal{U}$. This
implies that $S\left( y\right) \in U\subset \Sigma $ and therefore $y\in V$.
This concludes the proof.
\end{proof}

A classical theorem due to Whitehead guarantees the existence of convex
normal neighbourhoods $V$ at any point $x\in M$, (see \cite[Ch. 5]{On83} for
a treatment of this result in Lorentz manifolds). If $V\subset M$ is an open
convex normal neighbourhood and $x,y\in V$, then there exists a unique
geodesic segment joining $x$ and $y$.  Let us consider the open set $U=S\left( V\right) = \{ S(x) \mid x\in V \}$,
then for every $S(x) = X\neq Y = S(y)\in U$ and $\gamma \in X\cap Y$ verifying $T_{\gamma
}X\cap T_{\gamma }Y\neq \left\{ 0\right\} $ there exist a Jacobi field $J$
such that $J\left( s_{0}\right) =J\left( s_{1}\right) =0$ where $x=\gamma
\left( s_{0}\right) $ and $y=\gamma \left( s_{1}\right) $, but that is
not possible in a convex normal neighbourhood $V$ (see \cite[Prop. 10.10]%
{On83}).
So, in this case we have that $X=Y$ and the next definition is justified.

\begin{definition}\label{normal_neigh}
\label{normal} An open set $U\subset \Sigma$ in the reconstructive topology is called \emph{normal} if for
every $X,Y\in U $ and every $\gamma \in X\cap Y$ such that $T_{\gamma }X\cap
T_{\gamma }Y\neq \left\{ 0\right\} $ implies that $X=Y$.
\end{definition}

All the convex normal neighbourhoods at $x\in M$ set up a basis for the
topology of $M$ at $x$, then by lemma \ref{lemma3.3}, all the normal
neighbourhoods also constitute a basis for the topology of $\Sigma $.

Normal neighborhoods are not good enough to construct a differentiable
structure on $\Sigma$.  The following definition states the condition 
that will be required on open sets of $\Sigma$ to define a smooth atlas. If $N$ is manifold, we denote by $\widehat{T}N$
its reduced tangent bundle, this is, $\widehat{T}N=\cup _{x\in N}\widehat{T}%
_{x}N $ where $\widehat{T}_{x}N=T_{x}N-\left\{ 0\right\} $.

\begin{definition}\label{regular_neigh}
\label{definition3.6} A normal open set $U\subset \Sigma $ is said to be a
\emph{regular} open set if $U$ verifies that $\widehat{U}=\bigcup\limits_{X%
\in U}\widehat{T}X \subset T\mathcal{N}$ is a regular submanifold of $\widehat{T}\mathcal{U}$,
where $\mathcal{U=}\bigcup\limits_{X\in U}X$.
\end{definition}

We will prove that regular open sets constitute a basis for the
reconstructive topology of $\Sigma $.

\begin{theorem}
\label{theorem1} For every $X\in \Sigma $ there exists a regular open
neighbourhood $U\subset \Sigma $ of $X$.
\end{theorem}

\begin{proof}
Let $V\subset M$ be a relatively compact, globally hyperbolic, causally
convex normal neighbourhood of $q \in M$ and $U=S\left( V \right) \subset
\Sigma $ be the normal neighbourhood of $Q=S\left( q \right)$, in the sense
of Def. \ref{normal}, image of $V$ under the sky map $S$. We will use the
local coordinate chart $\psi \colon \mathcal{U} \to \mathbb{R}^{2m-3}$
described by eq. \eqref{equation3.4} on $\mathcal{U}$, with $\mathcal{U} =
\bigcup_{X\in U}X = \bigcup_{x \in V} S(x)$. Without any lack of generality,
because of the properties of $V$, we can assume the existence of a
coordinate chart $\varphi=\left( x^{1},\ldots ,x^{m}\right)$ and a
orthonormal frame $\lbrace E_{1},\ldots ,E_{m} \rbrace$ in $V$ such that the
map $\overline{\varphi}\colon \widehat{U} \to \mathbb{R}^{4m-3}$ (actually
we may use the same orthonormal frame $\lbrace E_{1},\ldots ,E_{m} \rbrace$
and coordinate chart $\varphi$ used to construct the coordinates $\overline{%
\psi}=\left(\mathbf{x},\mathbf{u},\mathbf{w},\mathbf{v}\right)$ of $T%
\mathcal{N}$) given by:
\begin{equation*}
\overline{\varphi}(J) = \left(x,u;v\right)=\left( x^{1},\ldots ,x^{m},\left[
u^{1},\dots ,u^{m} \right],\left\langle v^{1}\dots ,v^{m}\right\rangle
\right) \in \mathbb{R}^{3m-4}
\end{equation*}
defines a coordinate chart for $\widehat{U}=\bigcup\limits_{X\in U}\widehat{T%
}X$ in an analogous way to the chart $\overline{\psi}$ in \eqref{equation3.5}%
, where $J^{\prime }_{0}=\sum\limits_{j=1}^{m}v^{j}E_{j}\left( x\right) $
and again $v=\left\langle v^{1}\dots
,v^{m}\right\rangle=\left(v^1,\dots,v^m\right)\left(\mathrm{mod }\gamma
^{\prime }\right)$. Notice that because of eq. \eqref{tangent_sky} if $J$ is
tangent to a sky $S(q)$, $\gamma(0) = q$, then $J(0) = 0$, hence the local
chart $\overline{\varphi}$ is just the chart $\overline{\psi}$ setting $%
\mathbf{w} = 0$.

We will show now that the map $\overline{\varphi}$ gives a differentiable
structure to $\widehat{U}$ which does not depend on the chart $\varphi$ nor
the orthonormal frame chosen in $V$.

\begin{enumerate}
\item First, we will prove that the inclusion $i:\widehat{U}\hookrightarrow T%
\mathcal{U}\subset T\mathcal{N}$ is differentiable.   By construction of the
coordinates $\left( x,u\right) $ of $\widehat{U}$ and $\left( \mathbf{x},%
\mathbf{u}\right) $ of $T\mathcal{N}$ from the coordinates of $\mathbb{PN}^{+}%
\left( V\right) $ and $\mathbb{PN}^{+}\left( C\right) $ in eqs. \eqref{equation3.3} and \eqref{equation3.4} respectively, we have show that $\left. \sigma
\right\vert _{\mathbb{PN}^{+}\left( C\right) }:\mathbb{PN}^{+}\left( C\right)
\rightarrow \mathcal{U}$ is a diffeomorphism and therefore $\mathbf{x}\left(
x,u\right) $ and $\mathbf{u}\left( x,u\right) $ are differentiable functions
since they are the equations in coordinates of the submersion
\begin{equation*}
\sigma _{VC}=\left. \sigma \right\vert _{\mathbb{PN}^{+}\left( C\right)
}^{-1}\circ \left. \sigma \right\vert _{\mathbb{PN}^{+}\left( V\right) }:\mathbb{%
PN}^{+}\left( V\right) \mapsto \mathbb{PN}^{+}\left( C\right) \, .
\end{equation*}%
If $\mathbf{x}=\left( \mathbf{x}^{2},\dots ,\mathbf{x}^{m}\right) $, we
will denote $\left( 0,\mathbf{x}\right) =\left( 0,\mathbf{x}^{2},\dots ,%
\mathbf{x}^{m}\right) $. Consider then
\begin{equation*}
p\left( x,u\right) =\varphi ^{-1}\left( 0,\mathbf{x}\left( x,u\right)
\right) \in C\subset V
\end{equation*}%
and
\begin{equation*}
W\left( x,u\right) =E_{1}\left( p\left( x,u\right) \right) +\mathbf{u}%
^{2}\left( x,u\right) E_{2}\left( p\left( x,u\right) \right) +\dots +\mathbf{%
u}^{m}\left( x,u\right) E_{m}\left( p\left( x,u\right) \right)
\end{equation*}%
where $\mathbf{u}\left( x,u\right) =\left( \mathbf{u}^{3}\left( x,u\right)
,\dots ,\mathbf{u}^{m}\left( x,u\right) \right) \in T_{p\left( x,u\right) }M$
and $\mathbf{u}^{2}=\sqrt{1-\left( \mathbf{u}^{3}\right) ^{2}-\dots -\left(
\mathbf{u}^{m}\right) ^{2}}$. For any $\left( x,u\right) $ we define the
following map
\begin{equation*}
h\left( t,x,u\right) =\mathrm{exp}_{p\left( x,u\right) }\left( tW\left(
x,u\right) \right)
\end{equation*}%
It is clear that $h$ is differentiable by composition of differentiable
maps, and for fixed $\left( x_{0},u_{0}\right) $ the curve $\gamma _{\left(
x_{0},u_{0}\right) }\left( t\right) =h\left( t,x_{0},u_{0}\right) $ is a
null geodesic such that $\gamma _{\left( x_{0},u_{0}\right) }\left( 0\right)
\in C$. 
For any of these geodesics, we have the initial value problem of
Jacobi fields given by
\begin{equation}
\left\{
\begin{tabular}{l}
$J^{\prime \prime }=R\left( J,\gamma _{\left( x,u\right) }^{\prime }\right)
\gamma _{\left( x,u\right) }^{\prime }$ \\
$J\left( \tau \right) =0$ \\
$J^{\prime }\left( \tau \right) =\xi $%
\end{tabular}%
\ \right.  \label{JFDD}
\end{equation}%
where $R$ is the Riemann curvature tensor, $\tau $ is in the domain of $%
\gamma _{\left( x,u\right) }$ and $\xi \in T_{\gamma _{\left( x,u\right)
}\left( \tau \right) }M$.

If we express the Jacobi field $J$ as $J=\alpha ^{k}\partial / \partial x^{k}$, then %
eq. \eqref{JFDD} can be written as a system of differential equations
\begin{eqnarray*}
\frac{d^{2}\alpha ^{k}}{dt^{2}} &+&\frac{d\alpha ^{i}}{dt}\left( \Gamma
_{ij}^{k}\frac{\partial h^{j}}{\partial t}\right) +\alpha ^{i}\frac{d}{dt}%
\left( \Gamma _{ij}^{k}\frac{\partial h^{j}}{\partial t}\right) + \\
&+&\Gamma _{ln}^{k}\left( \frac{d\alpha ^{l}}{dt}+\Gamma _{ij}^{l}\alpha ^{i}%
\frac{\partial h^{j}}{\partial t}\right) \frac{\partial h^{n}}{\partial t}%
-\alpha ^{n}\frac{\partial h^{i}}{\partial t}\frac{\partial h^{j}}{\partial t%
}R_{jni}^{k}=0
\end{eqnarray*}%
for $k=1,\dots ,m$ where, for brevity, we write $\Gamma _{ij}^{k}=\Gamma
_{ij}^{k}\left( h\left( t,x,u\right) \right) $, $R_{jni}^{k}=R_{jni}^{k}%
\left( h\left( t,x,u\right) \right) $ and $h^{j}=x^{j}\circ h$.

If we transform this second order system into a first order one by using the
standard transformation $y^{k}=\alpha ^{k}$ and $y^{m+k}=d\alpha ^{k}/dt$
for $k=1,\dots ,m$ then, the system eq. \eqref{JFDD} has the form:
\begin{equation}
\left\{
\begin{array}{l}
\displaystyle{\frac{dy}{dt}=f\left( t,y,x,u\right) } \\
y\left( \tau \right) =\overline{\xi }%
\end{array}%
\right.  \label{JFFOSimple}
\end{equation}%
Let us denote as $y\left( t,x,u,\tau ,\overline{\xi }\right) $ the
solution of \ref{JFFOSimple}, corresponding to a Jacobi field $J_{\tau ,%
\overline{\xi }}\in \widehat{U}$ along the null geodesic $\gamma _{\left(
x,u\right) }$ with $J_{\tau ,\overline{\xi }}\left( \tau \right) =0$ and $ J'_{\tau ,\overline{\xi }}\left( \tau \right) = \xi$. 
By construction, for each $\left( x,u\right) $ there exists a unique $\tau $
such that $\varphi \left( h\left( \tau ,x,u\right) \right) =x$. We will
write this function as $\tau \left( x,u\right) $ and it is possible to show
easily that this $\tau $ is differentiable\footnote{It can be done applying the implicit function theorem fo the map $F\left( t,x,u\right) = \varphi (h\left( t,x,u\right)) -x$.}.
The solution $y\left( 0,x,u,\tau \left( x,u\right) ,%
\overline{\xi }\right) $ gives us the values of $J_{\tau ,\overline{\xi }%
}\left( 0\right) $ and $J_{\tau ,\overline{\xi }}^{\prime }\left( 0\right) $%
, and therefore it provides the coordinates $\mathbf{v}\left( x,u,v\right) $
and $\mathbf{w}\left( x,u,v\right) $. Because of the theorem on the regular
dependence of solutions of initial value problems with parameter (see for
instance \cite[chapter 5]{Ha64}), $y\left( 0,x,u,\tau \left( x,u\right) ,%
\overline{\xi }\right) $ is a differentiable function depending smoothly on
the data $\left( x,u,\overline{\xi }\right) $ and hence $\mathbf{v}\left(
x,u,v\right) $ and $\mathbf{w}\left( x,u,v\right) $ are differentiable
functions of $\left( x,u,v\right) $. This proves that $i:\widehat{U}%
\hookrightarrow T\mathcal{U}$ is differentiable.

\item The second step in this proof is to show that $i:\widehat{U}%
\hookrightarrow T\mathcal{U}$ is an immersion. For this purpose we will show
that any regular curve in $\widehat{U}$ is transformed by $i$ into a regular
curve in $T\mathcal{U}$. Let us consider a regular curve $c\left( s\right)
\in \widehat{U}$ with $s\in \left( -\varepsilon ,\varepsilon \right) $. This
means that $c\left( s\right) =J_{s}$ is a Jacobi field along a null
(parametrized) geodesic $\gamma _{s}$ verifying $J_{s}\left( t_{s}\right) =0$, 
and $J_{s}^{\prime }\left( t_{s}\right) =\xi \left( s\right) $ is not
proportional to $\gamma _{s}^{\prime }\left( t_{s}\right) $. We will prove
that $i_{\ast }\left( c^{\prime }\left( 0\right) \right) \neq 0$ if $%
c^{\prime }\left( 0\right) \neq 0$, that is
\begin{equation*}
c^{\prime }\left( 0\right) \neq 0\Rightarrow \left( i\circ c\right) ^{\prime
}\left( 0\right) \neq 0
\end{equation*}%
This curve $c$ can be written in coordinates as $\overline{\varphi }\left(
c\left( s\right) \right) =\left( x\left( s\right) ,u\left( s\right) ,v\left(
s\right) \right) $ with $\overline{\varphi }\left( c\left( 0\right) \right)
=\left( x_{0},u_{0},v_{0}\right) $ and it has a differentiable image in $T%
\mathcal{U}$. The inclusion $i$ transforms the coordinates of $c$ as
\begin{equation*}
\overline{\psi }\circ i\circ \left( \overline{\varphi }\right) ^{-1}\left(
x\left( s\right) ,u\left( s\right) ,v\left( s\right) \right) =
\end{equation*}%
\begin{equation*}
=\left( \mathbf{x}\left( x\left( s\right) ,u\left( s\right) \right) ,\mathbf{%
u}\left( x\left( s\right) ,u\left( s\right) \right) ,\mathbf{v}\left(
x\left( s\right) ,u\left( s\right) ,v\left( s\right) \right) ,\mathbf{w}%
\left( x\left( s\right) ,u\left( s\right) ,v\left( s\right) \right) \right)
\end{equation*}%
The map $\left( \mathbf{x}\left( x,u\right) ,\mathbf{u}\left( x,u\right)
\right) $ coincides with the map $\sigma _{VC}=\left. \sigma \right\vert _{%
\mathbb{PN}^{+}\left( C\right) }^{-1}\circ \left. \sigma \right\vert _{\mathbb{PN%
}^{+}\left( V\right) }:\mathbb{PN}^{+}\left( V\right) \mapsto \mathbb{PN}^{+}\left(
C\right) $ in coordinates, which is a submersion, then its differential has
maximal rank $2m-3$ and codimension $1$. If the curve with coordinates $%
\left( x\left( s\right) ,u\left( s\right) \right) $ is transversal to the
fibre of $\sigma _{VC}$ at $s=0$, then obviously $\left( i\circ c\right)
^{\prime }\left( 0\right) \neq 0$. In other case, we can take $c$ (defining $%
c^{\prime }\left( 0\right) $) as a regular curve verifying that $c\left(
s\right) =J_{s}$ lies on a fixed null geodesic $\gamma $, then
\begin{equation*}
\overline{\psi }\circ i\circ \left( \overline{\varphi }\right) ^{-1}\left(
x\left( s\right) ,u\left( s\right) ,v\left( s\right) \right) =\left( \mathbf{%
x}\left( x_{0},u_{0}\right) ,\mathbf{u}\left( x_{0},u_{0}\right) ,\mathbf{v}%
\left( x_{0},u_{0},v\left( s\right) \right) \mathbf{w}\left(
x_{0},u_{0},v\left( s\right) \right) \right)
\end{equation*}%
where $\left( \mathbf{x},\mathbf{u}\right) $ remains constant for every $s$.
Then the differential
\begin{equation*}
\left( d\mathbf{x}_{c\left( 0\right) }\left( c^{\prime }\left( 0\right)
\right) ,d\mathbf{u}_{c\left( 0\right) }\left( c^{\prime }\left( 0\right)
\right) \right) =\left( 0,0\right) .
\end{equation*}

This regular curve $c$ is a curve of Jacobi fields $J_{s}\in \widehat{U}$
along the null geodesic $\gamma $ such that $J_{s}\left( t_{0}+s\right) =0$
and $J_{s}^{\prime }\left( t_{0}+s\right) =\xi \left( s\right) $ for $s\in
\left( -\epsilon ,\epsilon \right) $ and hence $\xi \left( s\right) $ is a
vector field along $\gamma $ non-proportional to $\gamma ^{\prime }$ at $s=0$%
. We can assume, without any lack of generality that $t_{0}=0$ and the local
Cauchy surface $C$ associated to the chart $\overline{\psi }$ contains $%
\gamma \left( 0\right) $. We have that $J_{0}\left( 0\right) =0$. So,
\begin{equation*}
\left. \frac{d}{ds}\right\vert _{s=0}J_{s}\left( 0\right) =\lim_{s\mapsto 0}%
\frac{J_{s}\left( 0\right) -J_{0}\left( 0\right) }{s}=\lim_{s\mapsto 0}\frac{%
J_{s}\left( 0\right) }{s}
\end{equation*}%
By \cite[Prop. 10.16]{BE96}, we have that $J_{s}\left( t\right) =\left(
\mathrm{exp}_{\gamma \left( s\right) }\right) _{\ast }\left( \left(
t-s\right) \tau _{\left( t-s\right) \gamma ^{\prime }\left( s\right)
}J_{s}^{\prime }\left( s\right) \right) $ where for $v\in T_{\gamma
\left( s\right) }M$, the map $\tau _{v}:T_{\gamma \left( s\right) }M\rightarrow
T_{v}T_{\gamma \left( s\right) }M$ is the canonical isomorphism. Then
\begin{equation*}
\left. \frac{d}{ds}\right\vert _{s=0}J_{s}\left( 0\right) =\lim_{s\mapsto 0}%
\frac{1}{s}\left( \mathrm{exp}_{\gamma \left( s\right) }\right) _{\ast
}\left( \left( -s\right) \tau _{\left( -s\right) \gamma ^{\prime }\left(
s\right) }\xi \left( s\right) \right) =
\end{equation*}%
\begin{equation*}
=\lim_{s\mapsto 0}\left( \mathrm{exp}_{\gamma \left( s\right) }\right)
_{\ast }\left( \left( \frac{-s}{s}\right) \tau _{\left( -s\right) \gamma
^{\prime }\left( s\right) }\xi \left( s\right) \right) =\lim_{s\mapsto
0}\left( \mathrm{exp}_{\gamma \left( s\right) }\right) _{\ast }\left( -\tau
_{\left( -s\right) \gamma ^{\prime }\left( s\right) }\xi \left( s\right)
\right) =
\end{equation*}%
\begin{equation*}
=\left( \mathrm{exp}_{\gamma \left( 0\right) }\right) _{\ast }\left( -\tau
_{0}\xi \left( 0\right) \right) =-\xi \left( 0\right)
\end{equation*}%
Hence, we state that
\begin{equation*}
\left. \frac{d}{ds}\right\vert _{s=0}J_{s}\left( 0\right) =-\xi \left(
0\right)
\end{equation*}%
Since $\xi \left( 0\right) $ is not proportional to $\gamma ^{\prime }\left(
0\right) $ , then $d\mathbf{w}_{c\left( 0\right) }\left( c^{\prime }\left(
0\right) \right) \neq 0$, and this implies that $i\circ c$ is a regular
curve for $s=0$. Therefore $i$ is an immersion.

\item In the last step of this proof, we will show that $\widehat{U}\subset T%
\mathcal{U}$ is a regular submanifold. Let us consider the system of
ordinary differential equations \ref{JFFOSimple} for Jacobi fields in $%
\widehat{U}$. We will denote its solution as $y\left( t,x,u,\tau ,\overline{%
\xi }\right) $. If the origin of the parameter $t$ of \ref{JFFOSimple} is
lying in the local Cauchy surface $C$, we can write de Jacobi field $J$ such
that $J\left( \tau \right) =0$ and $J^{\prime }\left( \tau \right) =\xi $ as
the solution $y\left( t,\mathbf{x},\mathbf{u},\tau ,\overline{\xi }\right) $%
, where $\mathbf{x}=\left( 0,x^{2},\dots ,x^{m}\right) $ which can be
identified with the adapted coordinates $\mathbf{x}$ to $C$ in \ref%
{equation3.4}. Then, the pair $\left( \mathbf{x},\mathbf{u}\right) $ are the
coordinates of a point in $\mathbb{PN}^{+}\left( C\right) $ and therefore, they
determine the null geodesic $\gamma _{\left( \mathbf{x},\mathbf{u}\right) }$%
. In fact, $y\left( \tau ,\mathbf{x},\mathbf{u},\tau ,\overline{\xi }\right)
$ corresponds to the values $J\left( \tau \right) =0$ and $J^{\prime }\left(
\tau \right) =\xi $. Moreover, $y\left( 0,\mathbf{x},\mathbf{u},\tau ,%
\overline{\xi }\right) $ represents the values $J\left( 0\right) $ and $%
J^{\prime }\left( 0\right) $ which are lying in $C$, therefore $y\left( 0,%
\mathbf{x},\mathbf{u},\tau ,\overline{\xi }\right) $ is equivalent to give
the coordinates $\overline{\psi }\left( J\right) =\left( \mathbf{x},\mathbf{u%
},\mathbf{v},\mathbf{w}\right) $ of $J$ in $T\mathcal{N}$. Since $V$ is
relatively compact and due to the existence of flow boxes of non-vanishing
differentiable vector fields, we can assume, without any lack of generality,
that there exist a compact interval $I$ neighbourhood of $0$ such that the
parameter of any null geodesic defined by $\eta =E_{1}\left( p\right)
+u^{2}E_{2}\left( p\right) +\dots +u^{m}E_{m}\left( p\right) \in \mathbb{N}^{+}
_{p}\left( V\right) $ with $p\in V$ through $V$ is defined for $t\in I$.
Now, let us consider an arbitrary sequence $\{J_{n}\}\subset \widehat{U}%
\subset T\mathcal{N}$ converging to $J_{\infty }\in \widehat{U}\subset T%
\mathcal{N}$ in $T\mathcal{N}$. Proving that $\{J_{n}\}$
converges to $J_{\infty }$ in $\widehat{U}$ is sufficient to show that
$\widehat{U}\subset T\mathcal{U}$ is a regular submanifold.

The Jacobi fields $J_{n}$ and $J_{\infty }$ are fields along the null
geodesics $\gamma _{\left( \mathbf{x}_{n},\mathbf{u}_{n}\right) }$ and $%
\gamma _{\left( \mathbf{x}_{\infty },\mathbf{u}_{\infty }\right) }$
respectively and moreover there exist $t_{n},t_{\infty }\in I$ such that $%
J_{n}\left( t_{n}\right) $ and $J_{\infty }\left( t_{\infty }\right) $ are
proportional to $\gamma _{\left( \mathbf{x}_{n},\mathbf{u}_{n}\right)
}^{\prime }\left( t_{n}\right) $ and $\gamma _{\left( \mathbf{x}_{\infty },%
\mathbf{u}_{\infty }\right) }^{\prime }\left( t_{\infty }\right) $
respectively for every $n\in \mathbb{N}^{+}$. If their coordinates in $T\mathcal{%
N}$ are $\overline{\psi }\left( J_{n}\right) =\left( \mathbf{x}_{n},\mathbf{u%
}_{n},\mathbf{v}_{n},\mathbf{w}_{n}\right) $ and $\overline{\psi }\left(
J_{\infty }\right) =\left( \mathbf{x}_{\infty },\mathbf{u}_{\infty },\mathbf{%
v}_{\infty },\mathbf{w}_{\infty }\right) $ respectively, then we have that
\begin{equation*}
\lim_{n\mapsto \infty }\overline{\psi }\left( J_{n}\right) =\overline{\psi }%
\left( J_{\infty }\right)
\end{equation*}%
or equivalently
\begin{equation*}
\lim_{n\mapsto \infty }y\left( 0,\mathbf{x}_{n},\mathbf{u}_{n},t_{n},%
\overline{\xi }_{n}\right) =y\left( 0,\mathbf{x}_{\infty },\mathbf{u}%
_{\infty },t_{\infty },\overline{\xi }_{\infty }\right)
\end{equation*}%
Again because of the theorem on the regular dependence of solutions of
initial value problems with parameters, the solution $y\left( t,\mathbf{x},%
\mathbf{u},\tau ,\overline{\xi }\right) $ differentiably depends on the
variables $\left( t,x,u,\tau ,\overline{\xi }\right) $, therefore
\begin{equation*}
\lim_{n\mapsto \infty }y\left( t,\mathbf{x}_{n},\mathbf{u}_{n},t_{n},%
\overline{\xi }_{n}\right) =y\left( t,\mathbf{x}_{\infty },\mathbf{u}%
_{\infty },t_{\infty },\overline{\xi }_{\infty }\right)
\end{equation*}%
This implies that
\begin{equation*}
\lim_{n\mapsto \infty }J_{n}\left( t\right) =J_{\infty }\left( t\right)
\end{equation*}%
Since $I$ is compact, the sequence $\{t_{n}\}\subset I$ has a convergent
subsequence, so we can assume that $\{t_{n}\}$ itself verifies that $%
\lim_{n\mapsto \infty }t_{n}=\overline{t}\in I$. Then we have that
\begin{equation*}
\lim_{n\mapsto \infty }y\left( t_{n},\mathbf{x}_{n},\mathbf{u}_{n},t_{n},%
\overline{\xi }_{n}\right) =y\left( \overline{t},\mathbf{x}_{\infty },%
\mathbf{u}_{\infty },t_{\infty },\overline{\xi }_{\infty }\right)
\end{equation*}%
hence
\begin{equation*}
\begin{tabular}{l}
$\lim_{n\mapsto \infty }J_{n}\left( t_{n}\right) =J_{\infty }\left(
\overline{t}\right) $ \\
$\lim_{n\mapsto \infty }{J_{n}^{\prime }\left( t_{n}\right) }=J_{\infty
}^{\prime }\left( \overline{t}\right) $%
\end{tabular}%
\end{equation*}%
Since $J_{n}\left( t_{n}\right) $ is proportional to $\gamma _{\left(
\mathbf{x}_{n},\mathbf{u}_{n}\right) }^{\prime }\left( t_{n}\right) $ for
every $n\in \mathbb{N}^{+}$, then $J_{\infty }\left( \overline{t}\right) $ is
also proportional to $\gamma _{\left( \mathbf{x}_{\infty },\mathbf{u}%
_{\infty }\right) }^{\prime }\left( t_{\infty }\right) $, but $\gamma
_{\left( \mathbf{x}_{\infty },\mathbf{u}_{\infty }\right) }^{\prime }$ is a
null geodesic without conjugate points, therefore $\overline{t}=t_{\infty }$%
. This gives us
\begin{equation*}
\lim_{n\mapsto \infty }{J_{n}^{\prime }\left( t_{n}\right) }=J_{\infty
}^{\prime }\left( t_{\infty }\right)
\end{equation*}%
Recall that the coordinates of $\widehat{U}$ are given by $\overline{\varphi
}=\left( x,u,v\right) $ where $\varphi =\left( x^{1},\ldots ,x^{m}\right) $
is the chart in $V$. Then
\begin{equation*}
\lim_{n\mapsto \infty }\overline{\varphi }\left( J_{n}\right)
=\lim_{n\mapsto \infty }\left( \varphi \left( \gamma _{\left( \mathbf{x}_{n},%
\mathbf{u}_{n}\right) }\left( t_{n}\right) \right) ,\left[ \gamma _{\left(
\mathbf{x}_{n},\mathbf{u}_{n}\right) }^{\prime }\left( t_{n}\right) \right]
,\langle J_{n}^{\prime }\left( t_{n}\right) \rangle \right) =
\end{equation*}%
\begin{equation*}
=\left( \varphi \left( \gamma _{\left( \mathbf{x}_{\infty },\mathbf{u}%
_{\infty }\right) }\left( t_{\infty }\right) \right) ,\left[ \gamma _{\left(
\mathbf{x}_{\infty },\mathbf{u}_{\infty }\right) }^{\prime }\left( t_{\infty
}\right) \right] ,\langle J_{\infty }^{\prime }\left( t_{\infty }\right)
\rangle \right) =\overline{\varphi }\left( J_{\infty }\right)
\end{equation*}%
So, the sequence $\{J_{n}\}$ converges to $J_{\infty }$ in $\widehat{U}$.
This completes the proof.
\end{enumerate}
\end{proof}

\begin{corollary}
\label{corollary1} Regular open sets constitute a basis for the topology of $%
\Sigma $.
\end{corollary}

\begin{proof}
Let $W\subset \Sigma $ be any neighbourhood of $X\in \Sigma $. By theorem %
\ref{theorem1}, there exists a regular open neighbourhood $U\subset \Sigma $
of $X$. Then for any connected normal open set $V\subset U$, we have that $%
\widehat{V}\subset \widehat{U}$ and since $\widehat{U}$ is a regular
submanifold of $\widehat{T}\mathcal{U}$ then $\widehat{V}$ is a regular
submanifold of $\widehat{T}\mathcal{V}$ hence $V$ is a regular open set.
Therefore, any connected $V\subset W\cap U$ containing $X$ is a regular open
neighbourhood of $X$ such that $X\in V\subset W$.
\end{proof}

\begin{theorem}
\label{theorem2} Let $V\subset M$ be a globally hyperbolic convex normal
open set such that $U=S\left( V\right) \subset \Sigma $ is a regular open
set. Then $U$ has a canonical differentiable structure depending only on $%
\mathcal{N}$. Moreover, the restricted sky map $S:V\rightarrow U$ is a
diffeomorphism.
\end{theorem}

\begin{proof}
Any $X\in U$ is a regular submanifold of $\mathcal{N}$, therefore $\widehat{T%
}X$ is a regular submanifold of $\widehat{T}\mathcal{N}$. Denote $\widetilde{%
U}=\{\widetilde{X}=\widehat{T}X:X\in U\}$ and define the map $\widetilde{S}%
:V\rightarrow \widetilde{U}$ given by $\widetilde{S}\left( x\right) =%
\widetilde{S\left( x\right) }$. By definition \ref{definition3.6}, $\widehat{%
U}$ is a regular submanifold of $\widehat{T}\mathcal{U}$ which is an open
set of $\widehat{T}\mathcal{N}$ and since $\widehat{U}=\bigcup\limits_{X\in
U}\widehat{T}X$ then $\widehat{U}$ is foliated by $\{\widehat{T}X:X\in U\}$,
i.e. by $\widetilde{U}$. Denoting the distribution induced by that foliation
as $\mathcal{D}$ , we have that $\widetilde{U}=\widehat{U}/\mathcal{D}$,
hence $\widetilde{S}:V\rightarrow \widetilde{U}$ is a difeomorphism.
Moreover, by normality of $U$ then the map $U\rightarrow \widetilde{U}$
defined by $X\mapsto \widetilde{X}$ is a bijection, and it allows to
identify $U$ with $\widetilde{U}$. Therefore $U$ inherits from $\widetilde{U}
$ its structure of differentiable manifold and this implies that $%
S:V\rightarrow U$ is a difeomorphism.
\end{proof}

An important consequence of corollary \ref{corollary1} and theorem \ref%
{theorem2} is that, since $\widehat{U}$ is a regular submanifold of $T%
\mathcal{N}$, then the differentiable structure given in $\widehat{U}$
coincides with the inherited from $T\mathcal{N}$ on $\widehat{U}$. This
allows us to disregard the differentiable structure built in $\widehat{U}$
from the one involving $M$, but considering it inherited from $T\mathcal{N}$%
. In this way, the differentiable structure of $U$ is inherited from $%
\widetilde{U}=\widehat{U}/\mathcal{D} $, and then the space--time $M$ is not
necessary to obtain a differentiable structure for $\Sigma$, because it is
canonically obtained from $\mathcal{N}$. So, in order to recover the
strongly causal manifold $M$ from $\mathcal{N}$ and $\Sigma$ in section \ref%
{TReconstruccion}, we will not need $M$ itself but only $\mathcal{N}$ and $%
\Sigma$ and their corresponding structures.

\begin{corollary}
\label{corollary3.12} There exists a unique differentiable structure in $%
\Sigma$ compatible with the differentiable structure of any regular open set
$U\subset \Sigma$ given in theorem \ref{theorem2}. Moreover both, the sky
map $S\colon M \rightarrow \Sigma$ and the parachute map $P\colon \Sigma \to
M$ are diffeomorphisms.
\end{corollary}

\begin{proof}
For every $X\in\Sigma$ there exists a regular open set $W\subset\Sigma$. If $%
x\in M$ verifies that $S\left(x\right)=X$, we can consider a globally
hyperbolic convex normal neighbourhood $V\subset M$ of $x$ such that $%
U=S\left(V\right)\subset W$. By corollary \ref{corollary1}, the set $U$ is
also a regular open set containing $X$, and therefore, by theorem \ref%
{theorem2} $S:V\rightarrow U$ is a local diffeomorphism in $X$. The
bijectivity of $S$ provides us the global diffeomorphism $S:M\rightarrow
\Sigma$.
\end{proof}

\section{The reconstruction theorem}

\label{TReconstruccion}

\begin{definition}
Let $\left( M,\mathcal{C}\right) $ and $\left( \overline{M},\overline{\mathcal{C}}\right) $ be two strongly causal manifolds and $\left(
\mathcal{N},\Sigma \right) $ and $\left( \overline{\mathcal{N}},\overline{%
\Sigma }\right) $ theirs corresponding pairs of spaces of  light rays and
skies. We say that a map $\phi :\mathcal{N}\rightarrow \overline{\mathcal{N}}
$ \emph{preserves skies} if $\phi \left( X\right) \in \overline{\Sigma }$
for any $X\in \Sigma $. Moreover, $\left( M,\mathcal{C}\right) $ is said
to be \emph{recoverable} if for $\left( \overline{\mathcal{N}},\overline{%
\Sigma }\right) $ corresponding to $\left( \overline{M},\overline{\mathcal{C}}\right) $ another strongly causal manifolds and $\phi :\mathcal{%
N}\rightarrow \overline{\mathcal{N}}$ a diffeomorphism preserving skies,
then the map
\begin{equation*}
\varphi =\overline{P}\circ \phi \circ S:M\rightarrow \overline{M}
\end{equation*}%
is a conformal diffeomorphism on its image, where $\overline{P}:\overline{%
\Sigma }\rightarrow \overline{M}$ is the parachute map to $\overline{M}$.
\end{definition}

\begin{lemma}
\label{lemma4.5} Let $\left( M,\mathcal{C}\right)$ and $\left( \overline{%
M},\overline{\mathcal{C}}\right)$ be strongly causal manifolds and let $%
\left( \mathcal{N},\Sigma\right) $ and $\left( \overline{\mathcal{N}},%
\overline{\Sigma}\right)$ be their corresponding pair of spaces of null
geodesics and skies. If $\phi :\mathcal{N} \rightarrow \overline{\mathcal{N}}
$ a diffeomorphism preserving skies then $\phi$ induces the map $\Phi
:\Sigma \rightarrow \overline{\Sigma}$ defined by $\Phi\left(X\right)=\phi%
\left(X\right)$ verifying $\Phi$ is injective, open and continuous.
\end{lemma}

\begin{proof}
Obviously, $\Phi$ is well defined and injective. To show $\Phi$ is an open
map, given an open set $U\subset \Sigma$, we will study the set $\overline{U}%
=\Phi\left(U\right)$. We have that $\mathcal{U}=\bigcup\limits_{X\in U}X$ is
open in $\mathcal{N}$ and, since $\phi$ is a diffeomorphism, then $\overline{%
\mathcal{U}}=\phi\left(\mathcal{U}\right)$ is an open set in $\mathcal{%
\overline{N}}$. Moreover
\begin{equation*}
\overline{\mathcal{U}}=\phi\left(\mathcal{U}\right)=\phi\left(\bigcup%
\limits_{X\in U}X\right)=\bigcup\limits_{X\in U}\phi\left(X\right)
\end{equation*}
and since $\Phi$ is injective and $\overline{U}=\Phi\left(U\right)$ we have
\begin{equation*}
\bigcup\limits_{X\in
U}\phi\left(X\right)=\bigcup\limits_{\phi\left(X\right)\in \overline{U}%
}\phi\left(X\right)=\bigcup\limits_{\overline{X}\in \overline{U}}\overline{X}
\end{equation*}
Then, by lemma \ref{lemma3.1}, $\overline{U}$ is open in $\overline{\Sigma}$
and therefore $\Phi$ is open. Finally, we will show that $\Phi$ is
continuous. Given an open set $\overline{U}\subset \overline{\Sigma}$, the
set $\overline{\mathcal{U}}=\bigcup\limits_{\overline{X}\in \overline{U}}%
\overline{X}$ is open in $\overline{\mathcal{N}}$. Denote $U=\Phi^{-1}\left(%
\overline{U}\right)=\lbrace X\in\Sigma: \phi\left(X\right)\in\overline{U}
\rbrace$. By being $\phi$ a diffeomorphism, then the set $\mathcal{U}%
=\phi^{-1}\left(\overline{\mathcal{U}}\right)$ is open. Moreover
\begin{equation*}
\mathcal{U}=\phi^{-1}\left(\overline{\mathcal{U}}\right)=
\phi^{-1}\left(\bigcup\limits_{\overline{X}\in \overline{U}}\overline{X}%
\right)=\bigcup\limits_{\overline{X}\in \overline{U}}\phi^{-1}\left(%
\overline{X}\right)=\bigcup\limits_{X\in U}X
\end{equation*}
Again, by lemma \ref{lemma3.1}, $U$ is open in $\Sigma$ and therefore $\Phi$
is continuous.
\end{proof}

Restricting the map $\Phi$ of lemma \ref{lemma4.5} to its image, $%
\Phi:\Sigma\rightarrow\Phi\left(\Sigma\right)$ then it is clear that $\Phi$
is bijective, open and continuous, hence is a homeomorphism. This
homeomorphism induces, in virtue of lemma \ref{lemma3.3} or corollary \ref%
{corollary3.12}, the homeomorphism $\varphi=\overline{P}\circ\Phi\circ S$
onto an open set of $\overline{M}$. So, we can assume, with no lack of
generality that $\overline{\Sigma}=\Phi\left(\Sigma\right)$ and $\overline{M}%
=\overline{P}\circ\Phi\left(\Sigma\right)$.

\begin{theorem}
\label{theorem4.6} Let $\left( M,\mathcal{C}\right) $ be a strongly
causal manifold, then $M$ is recoverable.
\end{theorem}

\begin{proof}
Let $\left( \overline{M},\overline{\mathcal{C}}\right) $ be another
strongly causal manifold with $\left( \overline{\mathcal{N}},\overline{%
\Sigma }\right) $ its corresponding spaces of  light rays and skies, such
that $\phi :\mathcal{N}\rightarrow \overline{\mathcal{N}}$ a diffeomorphism
verifying $\phi \left( \Sigma \right) =\overline{\Sigma }$. It is clear that
the differential $\phi _{\ast }:T\mathcal{N}\rightarrow T\overline{\mathcal{N%
}}$ is a diffeomorphism.
Consider $Q\in \Sigma $ and $\overline{Q}=\phi \left( Q\right) \in \overline{%
\Sigma }$. By theorem \ref{theorem1}, there exist regular neighbourhoods $%
U\subset \Sigma $ of $Q$ and $\overline{U}\subset \overline{\Sigma }$ of $%
\overline{Q}$ that, by corollary \ref{corollary1}, we can assume $\overline{U%
}=\Phi \left( U\right) $. Then $\phi \left( \mathcal{U}\right) =\overline{%
\mathcal{U}}$ with $\mathcal{U}=\bigcup\limits_{X\in U}X$ and $\overline{%
\mathcal{U}}=\bigcup\limits_{\overline{X}\in \overline{U}}\overline{X}$, and
hence, the restriction $\phi _{\ast }:\widehat{T}\mathcal{U}\rightarrow
\widehat{T}\overline{\mathcal{U}}$ is also a diffeomorphism and it can be
restricted again to $\phi _{\ast }:\widehat{U}\rightarrow \widehat{\overline{%
U}}$. Since
\begin{equation*}
\phi _{\ast }\left( \widehat{U}\right) =\phi _{\ast }\left( \bigcup\limits_{%
\overline{X}\in U}\widehat{T}X\right) =\bigcup\limits_{X\in U}\phi _{\ast
}\left( \widehat{T}X\right) =\bigcup\limits_{X\in U}\widehat{T}\phi \left(
X\right) =\widehat{\overline{U}}
\end{equation*}%
and the regularity of $U$ and $\overline{U}$, we have that $\widehat{U}$ and $%
\widehat{\overline{U}}$ are regular submanifolds of $\widehat{T}\mathcal{U}$
and $\widehat{T}\overline{\mathcal{U}}$ respectively. Then $\phi _{\ast }:%
\widehat{U}\rightarrow \widehat{\overline{U}}$ is a bijective restriction of
a diffeomorphism between two regular submanifolds of $\widehat{T}\mathcal{U}$
and $\widehat{T}\overline{\mathcal{U}}$, then $\phi _{\ast }:\widehat{U}%
\rightarrow \widehat{\overline{U}}$ is a diffeomorphism. Denoting by $\mathcal{D}=\{\widehat{T}X:X\in U\}$, and
$\overline{\mathcal{D}}=\{\widehat{T}\overline{X}:\overline{X}\in \overline{U%
}\}$ the distributions in $\widehat{U}$ and $\widehat{\overline{U}}$, we see
that $\phi _{\ast }\mathcal{D=}\overline{\mathcal{D}}$. Therefore $\phi
_{\ast }:\widehat{U}\rightarrow \widehat{\overline{U}}$ induces a smooth map
$\widehat{U}/\mathcal{D}\overset{\phi _{\ast }}{\rightarrow }\widehat{%
\overline{U}}/\overline{\mathcal{D}}$ and we have the following commutative
diagram%
\begin{equation*}
\begin{tabular}{ccc}
$\widehat{U}$ & $\overset{\phi _{\ast }}{\rightarrow }$ & $\widehat{%
\overline{U}}$ \\
$\downarrow $ &  & $\downarrow $ \\
$\widehat{U}/\mathcal{D}$ & $\overset{\phi _{\ast }}{\rightarrow }$ & $%
\widehat{\overline{U}}/\overline{\mathcal{D}}$ \\
$\downarrow $ &  & $\downarrow $ \\
$U$ & $\overset{\Phi }{\rightarrow }$ & $\overline{U}$%
\end{tabular}%
\end{equation*}%
(recall the proof of Theorem \ref{theorem2} to see that the lower vertical
arrows are diffeomorphisms). Therefore we conclude that $\Phi :U\rightarrow
\overline{U}$, and $\Phi :\Sigma \rightarrow \overline{\Sigma }$ are
diffeomorphisms. So, in virtue of corollary \ref{corollary3.12}, the map $%
\varphi =\overline{P}\circ \Phi \circ S:M\rightarrow \overline{M}$ is a
diffeomorphism. Now, we need to show that $\varphi $ maps  light rays of $%
M$ into  light rays of $\overline{M}$. We can consider all the null
geodesics in the skies of a given null geodesic $\gamma $, denoted as
\begin{equation*}
S\left( \gamma \right) =\{\beta \in \mathcal{N}:\exists \text{ }X\in \Sigma
\text{ such that }\gamma ,\beta \in X\}
\end{equation*}%
Then
\begin{equation*}
\Phi \left( S\left( \gamma \right) \right) =\phi \left( S\left( \gamma
\right) \right) =\{\phi \left( \beta \right) \in \overline{\mathcal{N}}%
:\exists \text{ }X\in \Sigma \text{ such that }\gamma ,\beta \in X\}
\end{equation*}%
and since $\phi $ is a diffeomorphism preserving skies
\begin{equation*}
\Phi \left( S\left( \gamma \right) \right) =\{\phi \left( \beta \right) \in
\overline{\mathcal{N}}:\exists \text{ }\Phi \left( X\right) \in \overline{%
\Sigma }\text{ such that }\phi \left( \gamma \right) ,\phi \left( \beta
\right) \in \Phi \left( X\right) \}
\end{equation*}%
therefore
\begin{equation*}
\Phi \left( S\left( \gamma \right) \right) =\overline{S}\left( \phi \left(
\gamma \right) \right)
\end{equation*}%
So, it implies $\varphi \left( \gamma \right) =\overline{P}\circ \Phi \circ
S\left( \gamma \right) =\overline{P}\circ \overline{S}\circ \phi \left(
\gamma \right) =\phi \left( \gamma \right) \in \overline{\mathcal{N}}$ is a
null geodesic. By \cite[section 3.2]{HE73}, $\varphi $ is a conformal
diffeomorphism.
\end{proof}

\section{Causality and Legendrian isotopies}
\label{section-leg-isot}

Let us recall first some basic concepts from contact geometry that we are going to relate to causality properties of
space--times.

Let $\left( Y,\mathcal{H}\right) $ be a co-oriented $\left(
2n-1\right) $--dimensional  contact manifold with contact distribution $\mathcal{H}=\text{ker }\alpha $
where $\alpha \in T^{\ast }Y$ is a contact 1--form which defines
the co-orientation.   A differentiable family $\{\Lambda _{s}\}_{s\in \left[
0,1\right] }$ of legendrian submanifolds is called a \emph{legendrian isotopy}.
It is possible to describe a legendrian isotopy by a parametrization $%
F:\Lambda _{0}\times \left[ 0,1\right] \rightarrow Y$ verifying $F\left(
\Lambda _{0}\times \{s\}\right) =\Lambda _{s}\subset Y$ where $s\in \left[
0,1\right] $.   Notice that we are assuming that the map $F_s \colon \Lambda_0 \to \Lambda_s$, given by
$F_s (\lambda) = F(s,\lambda)$ is a diffeomorphism for all $s \in [0,1]$.

\begin{definition}
A parametrization $F$ of a legendrian isotopy is
said to be \emph{non-negative} if $\left( F^{\ast }\alpha \right) \left( \frac{\partial
}{\partial s}\right) \geq 0$ and \emph{non-positive} if $\left( F^{\ast
}\alpha \right) \left( \frac{\partial }{\partial s}\right) \leq 0$.
\end{definition}

\begin{definition}
We will say that two legendrian isotopies are \emph{equivalent} if their
corresponding parametrizations $F,\widetilde{F}:\Lambda_0 \times \left[0,1%
\right] \rightarrow Y$ verify $F\left(\Lambda_0 \times \lbrace s\rbrace
\right)=\widetilde{F}\left(\Lambda_0 \times \lbrace s\rbrace \right)$ for
every $s\in \left[0,1\right]$.
\end{definition}

\begin{lemma}
\label{lemma00250} Let $F,\widetilde{F}:\Lambda _{0}\times \left[ 0,1\right]
\rightarrow Y$ be two parametrizations of a legendrian isotopy $\{\Lambda
_{s}\}_{s\in \left[ 0,1\right] }$. If $F$ is non-negative
(respectively non-positive) then so is $\widetilde{F}$.
\end{lemma}

\begin{proof}
Let us consider a legendrian isotopy $\lbrace \Lambda_s \rbrace_{s\in\left[%
0,1\right]}$ given by two parametrizations $F,\widetilde{F}:\Lambda_0 \times %
\left[0,1\right] \rightarrow Y$.  Let us define the maps $F_s,\widetilde{F}%
_s:\Lambda_0 \rightarrow \Lambda_s \subset Y$ for $s\in\left[0,1\right]$ by $%
F_s\left(\lambda\right)=F\left(\lambda,s\right)$ as before. Then we have that
\begin{equation*}
F\left(\lambda,s\right)=\widetilde{F}\left(\varphi\left(\lambda,s\right),s%
\right)
\end{equation*}
where $\varphi\left(\lambda,s\right)=\widetilde{F}_{s}^{-1}\circ
F\left(\lambda,s\right)$.   To check that $\varphi$ is differentiable,
consider the differentiable map $\Upsilon:\Lambda_0 \times \left[0,1\right]
\rightarrow \mathcal{N}\times \left[0,1\right]$ defined by $%
\Upsilon\left(z,s\right)=\left(\widetilde{F}\left(z,s\right),s\right)$ whose differential at any $\left(z,s\right)$ is given by:
\begin{equation*}
d\Upsilon_{\left(z,s\right)} = \left(
\begin{matrix}
d\widetilde{F}_{\left(z,s\right)} \\
\text{Id}_{s}%
\end{matrix}
\right) = \left(
\begin{matrix}
\left(d\widetilde{F}_s\right)_{z} & * \\
0 & \text{Id}_{s}%
\end{matrix}
\right)
\end{equation*}
and since $\widetilde{F}_{s}$ is a diffeomorphism, then $\left(d\Upsilon%
\right)_{\left(z,s\right)}$ is a isomorphism, therefore by the Inverse
Function Theorem, $\Upsilon$ is a local difeomorphism onto its image in $%
\left(z,s\right)$ and $\varphi$ can be written locally as:
\begin{equation*}
\varphi\left(z,s\right)= \pi \circ
\Upsilon^{-1}\left(F\left(z,s\right),s\right)
\end{equation*}
where $\pi:\Lambda_0 \times \left[0,1\right] \rightarrow \Lambda_0$ is the
canonical projection.

\medskip

Defining $\phi:\Lambda_0 \times \left[0,1\right] \rightarrow \Lambda_0
\times \left[0,1\right]$ as $\phi\left(\lambda,s\right)=
\left(\varphi\left(\lambda,s\right),s\right)$, we have
\begin{align}
dF_{\left(\lambda,s\right)}\left( \frac{\partial}{\partial s}%
\right)_{\left(\lambda,s\right)} &= d\left(\widetilde{F}\circ\phi\right)_{%
\left(\lambda,s\right)}\left( \frac{\partial}{\partial s}\right)_{\left(%
\lambda,s\right)} = d\widetilde{F}_{\left(\varphi\left(\lambda,s\right),s%
\right)}\left(d\phi_{\left(\lambda,s\right)}\left( \frac{\partial}{\partial s%
}\right)_{\left(\lambda,s\right)}\right) = \nonumber \\
&= d\widetilde{F}_{\left(\varphi\left(\lambda,s\right),s\right)}\left(\left(
\frac{\partial}{\partial s}\right)_{\left(\lambda,s\right)}+d\varphi_{\left(%
\lambda,s\right)}\left( \frac{\partial}{\partial s}\right)_{\left(\lambda,s%
\right)}\right) \, . \label{dFds}
\end{align}

Notice that $\alpha\left( d\widetilde{F}_{\left(\varphi\left(\lambda,s\right),s\right)}
d\varphi_{\left(\lambda,s\right)}\left( \partial/\partial s
\right)\right)=0$, since $d\widetilde{F}_{\left(\varphi\left(\lambda,s\right),s\right)}
d\varphi_{\left(\lambda,s\right)}\left( \partial / \partial s
\right)\in T_{\left(\varphi\left(\lambda,s\right),s\right)}\Lambda_s$
because $d\varphi_{\left(\lambda,s\right)}\left( \partial /\partial s%
\right)\in T_{\varphi\left(\lambda,s\right)}\Lambda_0$.
Now, applying $\alpha$ to both sides of eq. \eqref{dFds} we get:
\begin{equation*}
\alpha\left(dF_{\left(\lambda,s\right)}\left( \frac{\partial}{\partial s}%
\right)_{\left(\lambda,s\right)} \right)= \alpha\left(d\widetilde{F}%
_{\left(\varphi\left(\lambda,s\right),s\right)}\left( \frac{\partial}{%
\partial s}\right)_{\left(\lambda,s\right)}\right)
\end{equation*}
hence
\begin{equation*}
\left(F^{*}\alpha\right) \left(\frac{\partial}{\partial s}%
\right)=\alpha\left(F_{*}\left( \frac{\partial}{\partial s}\right) \right)=
\alpha\left(\widetilde{F}_{*}\left( \frac{\partial}{\partial s}%
\right)\right)= \left(\widetilde{F}^{*}\alpha\right) \left(\frac{\partial}{%
\partial s}\right)
\end{equation*}
therefore the sign of the parametrizations $F$ and $\widetilde{F}$ coincides.
\end{proof}

As it was discussed in the introduction we are interested in the study of legendrian isotopies in the space of null
geodesics $\mathcal{N}$ of a Lorentz manifold $M$.  Recall that, in this
case, the co-orientation is defined by using the criterion
that the sign of $J\left( \mathrm{mod}\gamma ^{\prime }\right) \in T_{\gamma
}\mathcal{N}$ is the sign of $\mathbf{g}\left( J,\gamma ^{\prime
}\right) $, which is unambiguously determined for vectors $J$ in the class $[J] = J + \mathcal{J}_{\mathrm{tan}}(\gamma)$, where $\gamma \in
\mathcal{N}$ and $\mathbf{g}\in \mathcal{C}$.

Again, because of the remark after eq. \eqref{tangent_sky} the sky $X_{0} = S(x_0)\in \Sigma $ for any $x_{0}\in M$ is a legendrian
submanifold of $\mathcal{N}$ diffeomorphic to $S_{0}=\lbrace \left[ u\right] :u\in
\mathbb{N}^{+}_{x_{0}} \rbrace = \mathbb{PN}^{+}_{x_{0}} \cong S^{m-2}$, then
given a legendrian isotopy $\{X_{s}\}_{s\in \left[ 0,1\right] }$ where $%
X_{s} $ is the sky of $x_{s}\in M$ for $s\in \left[ 0,1\right] $, a
parametrization $F$ for it can be found of the form:
\begin{equation*}
F:S_{0}\times \left[ 0,1\right] \rightarrow \mathcal{N}.
\end{equation*}

\begin{lemma}
Any differentiable curve $\mu :\left[ 0,1\right] \rightarrow M$ defines a
legendrian isotopy parametrized by the function $F^{\mu }:S_{0}\times \left[ 0,1\right] \rightarrow
\mathcal{N}$ given by:
\begin{equation*}
F^{\mu }\left( \left[ u\right] ,t\right) =\gamma _{\left[ u_{s}\right] }
\end{equation*}%
with $S_{0}=\lbrace\left[ u\right] :u\in \mathbb{N}^{+}_{\mu \left(
0\right) }\rbrace$ and $u_{s}\in \mathbb{N}^{+}_{\mu \left( s\right) }$ the parallel transport
of $u\in \mathbb{N}^{+}_{\mu \left( 0\right) }$ along $\gamma $.  Moreover $F^\mu$ is a legendrian
isotopy of skies and $F^\mu_s(S_0) = S(\mu(s))$.
\end{lemma}

\begin{proof}
Let $\mathbf{g} \in \mathcal{C}$ be a metric in the space--time $M$ and
let $\mathcal{P}:T_{\mu\left(0\right)}M\times\left[0,1\right]\rightarrow TM$
be the parallel transport with respect to the Levi--Civita connection defined by $\mathbf{g}$
along $\mu$ given by $\mathcal{P}%
\left(u,s\right)=u_s\in T_{\mu\left(s\right)}M$. It is widely known that $%
\mathcal{P}$ is differentiable and the map $\mathcal{P}_s:T_{\mu\left(0%
\right)}M\rightarrow T_{\mu\left(s\right)}M$ defined by $\mathcal{P}%
_s\left(u\right)=\mathcal{P}\left(u,s\right)$ is a linear isometry. Let us
also consider the submersion $p_{\mathbb{N}^{+}}:\mathbb{N}^{+}\rightarrow \mathcal{N%
}$ given by $p_{\mathbb{N}^{+}}\left(u\right)=\gamma_{\left[u\right]}$. By
composition of differentiable maps, $p_{\mathbb{N}^{+}}\circ \mathcal{P}$
is differentiable and because of the linearity of $\mathcal{P}$ it induces a map $F^\mu$
on the quotient space $\mathbb{PN}^+$.

Moreover, since $\mathcal{P}_{s}$ is a linear isometry, then
\begin{equation*}
\mathbf{g}\left( u_{s},u_{s}\right) =\mathbf{g}\left( u,u\right) =0, \qquad u \in \mathbb{N}^+
\end{equation*}%
for any metric $\mathbf{g}\in \mathcal{C}$, therefore $u_{s}\in \mathbb{N}^{+}%
_{\mu \left( s\right) }$ and $\mathcal{P}_{s}\left( \mathbb{N}^{+}_{\mu \left(
0\right) }\right) =\mathbb{N}^{+}_{\mu \left( s\right) }$. For $s\in \left[ 0,1%
\right] $ we have
\begin{align*}
F^{\mu }\left( S_{0}\times \{s\}\right) & =\{F^{\mu }\left( \left[ u\right]
,s\right) \in \mathcal{N}:u\in \mathbb{N}^{+}_{\mu \left( 0\right) }\}
= \{\gamma _{\left[ u_{s}\right] }\in \mathcal{N}:u\in \mathbb{N}^{+}_{\mu
\left( 0\right) }\}= \\
& =\{\gamma _{\left[ v\right] }\in \mathcal{N}:v\in \mathbb{N}^{+}_{\mu \left(
s\right) }\}
= S\left( \mu \left( s\right) \right)
\end{align*}%
Hence, $F^{\mu }$ is a legendrian isotopy.
\end{proof}

\begin{lemma}\label{lemma00300}
Let $F:S_0 \times \left[0,1\right] \rightarrow \mathcal{N}
$ be a legendrian isotopy such that $F\left(S_0 \times \lbrace
s\rbrace \right)= S\left(\mu\left(s\right)\right)\in \Sigma$. Then the curve
$\mu:\left[0,1\right]\rightarrow M$ is differentiable and $F$ is equivalent
to $F^{\mu}$.
\end{lemma}

\begin{proof}
Let us define the map $F_{s}:S_{0}\rightarrow S\left( \mu \left( s\right)
\right) \subset \mathcal{N}$ given by $F_{s}\left( z\right) =F\left(
z,s\right) $ for $s\in \left[ 0,1\right] $. It is clear that $F_{s}$ is
differentiable for any $s\in \left[ 0,1\right] $. Now, take any $z_{0}\in
S_{0}$ and $\xi \in T_{z_{0}}S_{0}$. Since $F$ and $F_{s}$ are
differentiable maps, then the curve
\begin{equation*}
j\left( s\right) =\left( dF_{s}\right) _{z_{0}}\left( \xi \right) \in
T_{F\left( z_{0},s\right) }S\left( \mu \left( s\right) \right)
\end{equation*}%
is also differentiable in $\widehat{T}\mathcal{N}$ and $j\left(
s\right) $ is a Jacobi field along the null geodesic $F\left( z_{0},s\right)
\in \mathcal{N}$ for each $s\in \left[ 0,1\right] $. Let $s_{0}\in \left[
0,1\right] $ and $U=S\left( V\right) $ be a regular open neighbourhood of $%
\mu \left( s_{0}\right) $ . Let $\left( \widehat{U},\overline{\varphi }%
=\left( x,u,v\right) \right) $ and $\left( V,\varphi =x\right) $ be
coordinate charts as in theorem \ref%
{theorem1}. Then, since $j$ is differentiable, and $\widehat{U}$ is a
neighbourhood of $j\left( s_{0}\right) $ in $\widehat{T}\mathcal{N}$ we
conclude that $j\left( s\right) \in \widehat{U}$ for $s$ close to $%
s_{0}$, is differentiable and $\mu \left( s\right) =\varphi ^{-1}\circ
x\left( j\left( s\right) \right) \in V$.
Therefore $\mu $ is differentiable.
\end{proof}

Now, we need a simple result on the geometry of causal vectors on Lorentz manifolds that we state as the following technical lemma.

\begin{lemma}
\label{lemma00350} Let $M$ be a Lorentz manifold and $p\in M$. If $v\neq 0$
is a vector in $T_p M$ verifying $\mathbf{g}\left(u,v\right)\geq 0$ for any $%
u\in \mathbb{N}^{+}_{p}$ future, then $v$ is causal past.
\end{lemma}

\begin{proof}
First, we will see that if $v\in T_p M$ is spacelike, then there exists $%
u\in T_p M$ null future verifying $\mathbf{g}\left(u,v\right) <0$. So, let $%
v\in T_p M$ be spacelike and take some $z\in T_p M$ timelike future, then
since $\mathbf{g}\left(z,z\right)<0$ and $\mathbf{g}\left(v,v\right)>0$, the
equation
\begin{equation*}
\mathbf{g}\left(z+\lambda v,z+\lambda v\right)= \mathbf{g}%
\left(z,z\right)+2\lambda \mathbf{g}\left(z,v\right) +\lambda^2\mathbf{g}%
\left(v,v\right)=0
\end{equation*}
has two solutions $\lambda_1 , \lambda_2$ due to $\left(2 \mathbf{g}%
\left(z,v\right)\right)^2 - 4 \mathbf{g}\left(z,z\right) \mathbf{g}%
\left(v,v\right)>0$. These solutions can be written as
\begin{equation*}
\lambda_1=-\frac{\mathbf{g}\left(z,v\right)}{\mathbf{g}\left(v,v\right)} +
\sqrt{\frac{\mathbf{g}\left(z,v\right)^2}{\mathbf{g}\left(v,v\right)^2}-%
\frac{\mathbf{g}\left(z,z\right)}{\mathbf{g}\left(v,v\right)}}
\end{equation*}
\begin{equation*}
\lambda_2=-\frac{\mathbf{g}\left(z,v\right)}{\mathbf{g}\left(v,v\right)} -
\sqrt{\frac{\mathbf{g}\left(z,v\right)^2}{\mathbf{g}\left(v,v\right)^2}-%
\frac{\mathbf{g}\left(z,z\right)}{\mathbf{g}\left(v,v\right)}}
\end{equation*}
For $i=1,2$, let $u_i=z+\lambda_i v$ be the corresponding null vectors. We have
that
\begin{equation*}
\mathbf{g}\left(u_i,v\right)=\mathbf{g}\left(z,v\right)+\lambda_i\mathbf{g}%
\left(v,v\right)=\left(-1\right)^{i+1} \mathbf{g}\left(v,v\right) \sqrt{%
\frac{\mathbf{g}\left(z,v\right)^2}{\mathbf{g}\left(v,v\right)^2}-\frac{%
\mathbf{g}\left(z,z\right)}{\mathbf{g}\left(v,v\right)}}
\end{equation*}
hence $\mathbf{g}\left(u_2,v\right)<0$.

Let us see now that $u_2$ is null future. Since
\begin{equation*}
\mathbf{g}\left(u_1,u_2\right)=2\left[\mathbf{g}\left(z,z\right)-\frac{%
\mathbf{g}\left(v,z\right)^2}{\mathbf{g}\left(v,v\right)}\right]<0
\end{equation*}
therefore $u_1$ and $u_2$ are in the same time--cone. Moreover
\begin{equation*}
\mathbf{g}\left(u_i,z\right)=\mathbf{g}\left(v,v\right)\left[ \frac{\mathbf{g%
}\left(z,z\right)}{\mathbf{g}\left(v,v\right)}-\frac{\mathbf{g}%
\left(z,v\right)^2}{\mathbf{g}\left(v,v\right)^2} \right]\pm \sqrt{\frac{%
\mathbf{g}\left(z,v\right)^2}{\mathbf{g}\left(v,v\right)^2}-\frac{\mathbf{g}%
\left(z,z\right)}{\mathbf{g}\left(v,v\right)}}\mathbf{g}\left(z,v\right)
\end{equation*}
with the positive sign corresponding to $i=1$ and the negative to $i=2$. It
can be observed that if $\mathbf{g}\left(z,v\right)>0$ then $\mathbf{g}%
\left(u_2,z\right)<0$ therefore $u_2$ is in the same time--cone of $z$, hence $%
u_2$ is null future. In case of $\mathbf{g}\left(z,v\right)<0$ we have that $%
\mathbf{g}\left(u_1,z\right)<0$, then $u_1$ (and also $u_2$) is in the same
time--cone of $z$, therefore $u_1$ and $u_2$ are null future.

At this point, we have proven the equivalent result: If for any $u\in T_p M$
null future $\mathbf{g}\left(u,v\right)\geq 0$ is verified, then $v\in T_p M$
is causal. But if $v$ is causal future, then $\mathbf{g}\left(u,v\right)\leq
0$, hence $v=0$ contradicting the hypothesis, therefore $v$ must be causal
past.
\end{proof}

Let us recall that a curve $\mu :\left[ a,b\right] \rightarrow M$ is a \emph{null curve}
if it is differentiable and $\mathbf{g}\left( \mu ^{\prime },\mu ^{\prime
}\right) =0$. Notice that this is a conformal property and $\mu$ doesn't have to be a
regular curve.

\begin{definition}
The set of all null curves $\mu:I\rightarrow M$ will be
denoted as $\mathfrak{L}\left(M\right)$. The subset of $\mathfrak{L}%
\left(M\right)$ consisting of all time--orientable (future or past) null curves $\mu$ will be
denoted as $\mathfrak{L}_c\left(M\right)$, i.e., $\mu \in \mathfrak{L}_c\left(M\right)$ if $\mu$ is differentiable, $\mathbf{g}\left( \mu ^{\prime },\mu ^{\prime
}\right) =0$ and either $\mu'(s) \in \mathbb{N}^+$ for all $s$ or $\mu'(s) \in \mathbb{N}^-$  for all $s$.
\end{definition}

\begin{proposition}\label{prop00300} 
The curve $\mu $ is causal past (respectively causal
future) if and only if $F^{\mu }$ is a non-negative (respectively
non-positive) legendrian isotopy.
\end{proposition}

\begin{proof}
Let us suppose that $\mu$ is causal past. Since $F^{\mu}\left(\left[u\right]%
,s\right)=\gamma_{\left[u_s\right]}$ then giving parameters to the geodesics
$\gamma_{\left[u_s\right]}$ we can write
\begin{equation*}
F^{\mu}\left(\left[u\right],s\right)\left(t\right)=\gamma_{\left[u_s\right]%
}\left(t\right)=\mathrm{exp}_{\mu\left(s\right)}\left(tu_s\right)
\end{equation*}
which is a null geodesic variation of the null geodesic $\gamma_{\left[%
u_{s_0}\right]}$ for every $s_0 \in \left[0,1\right]$. By Lemma \ref%
{lemmaDC92}, we have that the Jacobi field $J_{s_0}\left(t\right)$ defined by
this geodesic variation verifies
that $J_{s_0}\left(0\right)=\mu ^{\prime }\left(s_0\right)$ and $J^{\prime
}_{s_0}\left(0\right)=\left.\frac{D}{ds}\right|_{s=s_0}u_s$, and since $u_s$
is the parallel transport of $u$ along $\mu$, then $J^{\prime
}_{s_0}\left(0\right)=0$. Hence, since
\begin{equation*}
F_{*}^{\mu}\left(\frac{\partial}{\partial s}\right)_{\left(\left[u\right]%
,s_0\right)}= \left.\frac{\partial}{\partial s}\right|_{\left(\left[u\right]%
,s_0\right)}F^{\mu}\left(\left[u\right],s\right)=\left.\frac{\partial}{%
\partial s}\right|_{\left(s_0,t\right)}\left(\mathrm{exp}_{\mu\left(s%
\right)}\left(tu_s\right)\right)=J_{s_0}\left(t\right)
\end{equation*}
we have that
\begin{equation*}
\alpha\left(F_{*}^{\mu}\left(\frac{\partial}{\partial s}\right)
\right)_{\left(\left[u\right],s_0\right)} = \alpha\left(
J_{s_0}\left(t\right)\right)=\mathbf{g}\left(J_{s_0}\left(t\right),\gamma
^{\prime }_{\left[u_{s_0}\right]}\left(t\right)\right)=
\end{equation*}
\begin{equation*}
=\mathbf{g}\left(J_{s_0}\left(0\right),\gamma ^{\prime }_{\left[u_{s_0}%
\right]}\left(0\right)\right) =\mathbf{g}\left(\mu ^{\prime
}\left(s_0\right),u_{s_0}\right) \geq 0
\end{equation*}
since $\mu ^{\prime }\left(s_0\right)$ is causal past where it does not
vanish and $u_{s_0}$ null future. This shows that $F^{\mu}$ is a
non-negative legendrian isotopy.

Now, let us suppose that $F^{\mu}$ is non-negative. So, we have as before
\begin{equation*}
F^{\mu}\left(\left[u\right],s\right)\left(t\right)=\gamma_{\left[u_s\right]%
}\left(t\right)=\mathrm{exp}_{\mu\left(s\right)}\left(tu_s\right)
\end{equation*}
then if $\alpha\left(F_{*}^{\mu}\left(\frac{\partial}{\partial s}\right)
\right)_{\left(\left[u\right],s_0\right)}\geq 0$ for any $\left(\left[u%
\right],s_0\right)$, we have that
\begin{equation*}
0 \leq \alpha\left(F_{*}^{\mu}\left(\frac{\partial}{\partial s}\right)
\right)_{\left(\left[u\right],s_0\right)} = \mathbf{g}\left(\mu ^{\prime
}\left(s_0\right),u_{s_0}\right) .
\end{equation*}
Then because of Lemma \ref{lemma00350} we obtain that $\mu ^{\prime
}\left(s_0\right)$ is causal past provided that $\mu ^{\prime }\left(s_0\right)\neq 0$
for every $s_0 \in\left[0,1\right]$.
\end{proof}

\begin{corollary}
\label{cor00300} A legendrian isotopy of skies $\lbrace
S\left(\mu\left(s\right)\right) \rbrace_{s\in\left[0,1\right]}$ is non-negative
if and only if the curve $\mu:\left[0,1\right] \rightarrow M$ is causal past.
\end{corollary}

\begin{proof}
By Lemma \ref{lemma00300}, a legendrian isotopy of skies $F:S_{0}\times \left[ 0,1\right] \rightarrow
\mathcal{N}$ defines a differentiable curve $\mu :\left[ 0,1\right]
\rightarrow M$ such that $F$ is equivalent to $F^{\mu }$. By Lemma \ref%
{lemma00250}, $F^{\mu }$ is non-negative, then Proposition \ref{prop00300}
shows that every regular segment of $\mu $ is causal past, therefore $\mu $
is causal past because is the union of causal past segments.
\end{proof}

\section{Celestial curves and reconstruction theorem}

\begin{definition}
\label{definition4.1} A tangent vector $J\neq 0$ at $T_{\gamma}\mathcal{N}$
will be called a \emph{celestial vector} if there exists a sky $S\in \Sigma $
such that $J\in T_{\gamma}S\subset T\mathcal{N}$.
We will denote the set of all \emph{celestial vectors} by $\widehat{\Sigma }\subset T%
\mathcal{N}$  i.e. with the notation introduced in Section \ref{differentiable_structure}, $\widehat{\Sigma }%
=\bigcup\limits_{X\in \Sigma }\widehat{T}X\subset \widehat{T}\mathcal{N}$.

A differentiable curve $\Gamma :I\rightarrow \mathcal{N}$ is called a \emph{%
celestial curve} if $\Gamma ^{\prime }\left( s\right) \in \widehat{\Sigma }$
for every $s\in I$. We denote the set of celestial curves as $\mathfrak{C%
}\left( \mathcal{N}\right) $.
\end{definition}

\begin{lemma}
\label{lemma4.2} Let $\Gamma :\left[ a,b\right] \rightarrow \mathcal{N}$ be
a differentiable curve in $\mathcal{N}$ such that $\Gamma \left( s\right)
=\gamma_{s}\subset M$. Then there exists a geodesic variation $\mathbf{f}:%
\mathcal{W}_{0}\rightarrow M$, where $\mathcal{W}_{0}=\left\{ \left(
s,t\right) \in \left[ a,b\right]\times\mathbb{R}:t\in I_{s}\right\} $ and $%
I_{s}$ is an open neighbourhood of $0$, such that
\begin{equation*}
\mathbf{f}\left( s,t\right) =\gamma _{s}\left( t\right)
\end{equation*}%
for every $\left( s,t\right) \in \mathcal{W}_{0}$. Furthermore, $\mathcal{W}%
_{0}$ is open in $\left[ a,b\right]\times\mathbb{R}$.
\end{lemma}

\begin{proof}
Let us consider the following differentiable maps: the canonical projections $\sigma:\mathbb{N}^{+}\rightarrow \mathcal{N}$ and $\pi^{\mathbb{N}^{+}}_{M}:\mathbb{N}^{+}\rightarrow M$ and the exponential map $\mathrm{exp}:J\times \mathbb{N}^{+}\rightarrow M$ defined by $\mathrm{exp}\left(t,v\right)=\mathrm{exp}_{\pi_{M}^{\mathbb{N}^{+}}\left(v\right)}\left(tv\right)$ with $J\subset \mathbb{R}$ an interval containing $0\in \mathbb{R}$.
Given the differentiable curve $\Gamma:\left[ a,b\right]\rightarrow \mathcal{N}$, a lift of $\Gamma$ can be constructed in $\mathbb{N}^{+}$ by local sections of $\sigma$.
By compactness of $\Gamma$, it is elementary to check that there exist differentiable local sections $s_i:U_i\rightarrow \mathbb{N}^{+}$ and a sequence of intervals $\lbrace I_i \rbrace_{i=1,\ldots,n}$ such that $\lbrace U_i\rbrace_{i=1,\ldots,n}$ is a finite covering of $\Gamma$ in $\mathcal{N}$, $\lbrace I_i \rbrace_{i=1,\ldots,n}$ is a covering of $\left[ a,b\right]$ such that $\Gamma\left(I_i\right)\subset U_i$ for every $i=1,\ldots,n$ with nonempty intersections $\emptyset\neq \left(a_i , b_i\right)=I_i \cap I_{i+1}$ for every $i=1,\ldots,n-1$ where $a_i\in I_i$ and $b_i\in I_{i+1}$. 
The restriction to $\Gamma\left(I_i\right)$ of every corresponding section defines a curve $\beta_i=\left. s_i\right|_{\Gamma}:I_i\rightarrow \mathbb{N}^{+}$ such that $\pi^{\mathbb{N}^{+}}_{M}\left(\beta_i\right)=\alpha_i\subset M$.
We can define variations $\mathbf{x}_i$ in $M$ from the lifts $s_i$ as $\mathbf{x}_i\left(s,t_i\right)=\mathrm{exp}\left(t_i,\beta_i\left(s\right)\right)=\mathrm{exp}_{\alpha_i\left(s\right)}\left(t_i\beta_i\left(s\right)\right)$.
These variations run through  light rays of segments of the curve $\Gamma\subset \mathcal{N}$.
Moreover we have that
\begin{equation}
\begin{tabular}{ccccccl}
    &  $\left.\Gamma\right|_{I_i}$     &  & $\left. s_i\right|_{\Gamma}$ &  & $\mathrm{exp}\left(t_i,\cdot\right)$ & \\
$I_i$ & $\longrightarrow$ & $\mathcal{N}$ & $\longrightarrow$ & $\mathbb{N}^{+}$ & $\longrightarrow$ & $M$ \\
$s$   & $\mapsto$ & $\Gamma\left(s\right)$ & $\mapsto$ & $\beta_i\left(s\right)$ & $\mapsto$ & $\mathbf{x}_i\left(s,t\right)=\mathrm{exp}_{\alpha_i\left(s\right)}\left(t_i\beta_i\left(s\right)\right)$
\end{tabular}
\end{equation}
is a composition of differentiable maps, then the variations $\mathbf{x}_i$ are differentiable.
We need to glue in a differentiable way all the $\mathbf{x}_i$.
For every $i=1,\ldots,n-1$ there exist differentiable functions $\lambda,\tau:\left(a_i,b_i\right)\rightarrow \mathbb{R}$ with $\lambda\left(s\right)> 0$ for all $s\in\left(a_i,b_i\right)$ such that 
\begin{equation}
\label{lemma6-2-variation}
\mathbf{x}_i\left(s,t_i\right)=\mathbf{x}_{i+1}\left(s,t_{i+1}\right)=\mathbf{x}_{i+1}\left(s,\lambda\left(s\right)t_i+\tau\left(s\right)\right).
\end{equation}
Let us take $c_i,d_i\in\left(a_i,b_i \right)$ such that $c_i<d_i$ and a $C^{\infty}$ function $\varphi_i$ verifying $\varphi_i\left(s\right)=0$ if $s\leq c_i$ and $\varphi_i\left(s\right)=1$ if $s\geq d_i$.
Define, then, the curve $\beta$ as
\begin{equation}
\beta\left(s\right) = \left\{
\begin{array}{lcl}
\left(1-\varphi_i\left(s\right)\left( \frac{\lambda\left(s\right) - 1}{\lambda\left(s\right)} \right)\right)\cdot \frac{\partial \mathbf{x}_i }{\partial t_i}\left(s,-\varphi_i\left(s\right)\frac{\tau\left(s\right)}{\lambda\left(s\right)}\right) & & \text{if } s\in \left[c_i,d_i\right] \\
\beta_i\left(s\right) & & \text{if } s\in \left(d_{i-1},c_i\right)
\end{array}
\right.
\end{equation}
with $i=1,\ldots,n-1$. Clearly, $\beta$ verifies that $\sigma\left(\beta\left(s\right)\right)=\gamma_s = \Gamma\left(s\right)$.
The curve
\begin{equation}\label{lemma6-2-curve}
s \mapsto\left(1-\varphi_i\left(s\right)\left( \frac{\lambda\left(s\right) - 1}{\lambda\left(s\right)} \right)\right)\cdot \frac{\partial \mathbf{x}_i }{\partial t_i}\left(s,-\varphi_i\left(s\right)\frac{\tau\left(s\right)}{\lambda\left(s\right)}\right)\in \mathbb{N}^{+}
\end{equation}
is defined and differentiable for all $s \in \left(a_i,b_i\right)$. Moreover, by \ref{lemma6-2-variation}, we have that 
\[
\frac{\partial \mathbf{x}_i}{\partial t_i}\left(s,t_i\right)=\lambda\left(s\right)\frac{\partial \mathbf{x}_{i+1}}{\partial t_{i+1}}\left(s,\lambda\left(s\right)t_i + \tau\left(s\right)\right)
\]
Then, if $s>d_i$ we have that $\varphi\left(s\right)=1$ and $\beta_{i+1}\left(s\right)= \frac{\partial \mathbf{x}_{i+1}}{\partial t_{i+1}}\left(s,0\right)$ and therefore 
\[
\left(1-\varphi_i\left(s\right)\left( \frac{\lambda\left(s\right) - 1}{\lambda\left(s\right)} \right)\right)\cdot \frac{\partial \mathbf{x}_i }{\partial t_i}\left(s,-\varphi_i\left(s\right)\frac{\tau\left(s\right)}{\lambda\left(s\right)}\right) =
\left(1-\left( \frac{\lambda\left(s\right) - 1}{\lambda\left(s\right)} \right)\right)\cdot \frac{\partial \mathbf{x}_i }{\partial t_i}\left(s,-\frac{\tau\left(s\right)}{\lambda\left(s\right)}\right)=
\]
\[
=\frac{1}{\lambda\left(s\right)}\cdot \frac{\partial \mathbf{x}_i }{\partial t_i}\left(s,-\frac{\tau\left(s\right)}{\lambda\left(s\right)}\right)= 
\frac{\partial \mathbf{x}_{i+1}}{\partial t_{i+1}}\left(s,0\right)= \beta_{i+1}\left(s\right)
\]
where we are taking $t_i=-\frac{\tau\left(s\right)}{\lambda\left(s\right)}$ and hence $t_{i+1}=\lambda\left(s\right)\left(-\frac{\tau\left(s\right)}{\lambda\left(s\right)} \right)+ \tau\left(s\right)=0$. This implies that for any $s>d_i$ the curve of \ref{lemma6-2-curve} coincides with $\beta_{i+1}\left(s\right)$, and moreover, it is trivial to observe that also coincides with $\beta_i\left(s\right)$ for $s<c_i$. Then $\beta$ is differentiable.
Now, if we denote $\alpha\left(s\right)=\pi^{\mathbb{N}^{+}}_{M}\left(\beta\left(s\right)\right)$ we can define the required variation $\mathbf{f}\left(s,t\right)=\mathrm{exp}_{\alpha\left(s\right)}\left(t\beta\left(s\right)\right)$.

To prove $\mathcal{W}_{0}$ is open, consider the geodesic spray $X_{\mathbf{g%
}}\in \mathfrak{X}\left( TM\right) $ and choose any $\left(
s_{0},t_{0}\right) \in \mathcal{W}_{0}$. The curve $\mathbf{f}\left(
s_{0},t\right) $ is a null geodesic passing through $\mathbf{f}\left(
s_{0},t_{0}\right) $, then the curve in $TM$ given by $\left( \mathbf{f}%
\left( s_{0},t\right) ,\frac{\partial \mathbf{f}}{\partial t}\left(
s_{0},t\right) \right) \in TM$ is an integral curve of $X_{\mathbf{g}}$
passing through $\left( \mathbf{f}\left( s_{0},t_{0}\right) ,\frac{\partial
\mathbf{f}}{\partial t}\left( s_{0},t_{0}\right) \right) $. By \cite[Theorem
4.1.5]{Ab88}, there exists an open neighbourhood $U_{0}$ of $\left( \mathbf{f%
}\left( s_{0},t_{0}\right) ,\frac{\partial \mathbf{f}}{\partial t}\left(
s_{0},t_{0}\right) \right) \in TM$ and an open interval $I$ such that the
flow $F$ of $X_{\mathbf{g}}$ is defined in $U_{0}\times I$. Restricting $F$
to the set
\begin{equation*}
\left\{ \left( \mathbf{f}\left( s,t_{0}\right) ,\frac{\partial \mathbf{f}}{%
\partial t}\left( s,t_{0}\right) \right) \in TM:s\in \left[ a,b\right]
\right\} \cap U_{0}
\end{equation*}%
then there exists an open neighbourhood $H_{0}$ of $s_{0}\in \left[
a,b\right] $ such that
\begin{equation*}
\mathcal{K}=\left\{ \left( \mathbf{f}\left( s,t_{0}\right) ,\frac{\partial
\mathbf{f}}{\partial t}\left( s,t_{0}\right) \right) \in TM:s\in
H_{0}\right\}
\end{equation*}%
is totally contained in $U_{0}$ so that the flow $F$ is defined in $\mathcal{%
K}\times I$. Therefore, $\mathbf{f}\left( s,t\right) $ (and also $\frac{%
\partial \mathbf{f}}{\partial t}\left( s,t\right) $) is defined in $%
H_{0}\times I$. Since $H_{0}\times I$ is an open neighbourhood of $\left(
s_{0},t_{0}\right) $ and since by definition of $\mathbf{f}$, it is
contained in $\mathcal{W}_{0}$, then we can conclude that $\mathcal{W}_{0}$
is open in $\left[ a,b\right] \times \mathbb{R}$.
\end{proof}

\begin{proposition}
\label{proposition4.3} If the curve $\Gamma :\left[ 0,1\right]
\rightarrow \mathcal{N}$ with $\Gamma \left( s\right) =\gamma _{s}\in
\mathcal{N}$ is celestial then there exists a
null curve $\mu :\left[ 0,1\right] \rightarrow M$ such that $\gamma
_{s}\left( \tau \right) =\exp _{\mu \left( s\right) }\left( \tau \sigma
\left( s\right) \right) $ where $\sigma \left( s\right) \in \mathbb{N}^{+}_{\mu
\left( s\right) }$ is a differentiable curve proportional to $\mu ^{\prime
}\left( s\right) $ wherever $\mu $ is regular.
\end{proposition}

\begin{proof}
First, assume the existence of a celestial curve $\Gamma $ such that $\Gamma
\left( 0\right) =\gamma _{0}$ and $\Gamma \left( 1\right) =\gamma _{1}$.

By lemma \ref{lemma4.2} there exists a geodesic variation $\mathbf{f}:%
\mathcal{W}_{0}\rightarrow M$ where
\begin{equation*}
\mathcal{W}_{0}=\left\{ \left( s,t\right) \in \left[ 0,1\right] \times
\mathbb{R}:t\in I_{s}\right\}
\end{equation*}%
being $I_{s}$ the domain of the parametrization of $\gamma _{s}$ defined by $%
\mathbf{f}$ . Now, we want to prove that there exists a differentiable
function $t:\left[ 0,1\right] \rightarrow I$ such that for every $s\in \left[
0,1\right] $, the Jacobi field $J_{s}$ along $\gamma _{s}$ defined by $%
\mathbf{f}$ verifies
\begin{equation*}
J_{s}\left( t\left( s\right) \right) =\lambda _{s}\gamma _{s}^{\prime
}\left( t\left( s\right) \right) \in T_{\gamma _{s}\left( t\left( s\right)
\right) }M
\end{equation*}%
for some $\lambda _{s}\in \mathbb{R}$. By definition \ref{definition4.1}, $%
J_{s}$ must be proportional to $\gamma _{s}^{\prime }$ in some point $t_{s}$%
. In lemma \ref{lemma4.2}, we stated that $\mathcal{W}_{0}$ is open, then
for every $\left( s,t\right) \in \mathcal{W}_{0}$ there exist intervals $%
K_{s}$, $H_{s}$ such that $\left( s,t\right) \in K_{s}\times H_{s}\subset
\mathcal{W}_{0}$ where the geodesic variation $\mathbf{f}$ is defined.
Choose a pair $\left( s_{0},t_{0}\right) $ verifying $J_{s_{0}}\left(
t_{0}\right) =\lambda \gamma _{s_{0}}^{\prime }\left( t_{0}\right) $.
Without any lack of generality, we can consider $K_{s_{0}}\times H_{s_{0}}$
such that $S\left( \mathbf{f}\left( K_{s_{0}}\times H_{s_{0}}\right) \right)
\subset U$ where $U\subset \Sigma $ is a normal neighbourhood. Define the
set
\begin{equation*}
A_{s_{0}}=\left\{ \left( s,t\right) \in K_{s_{0}}\times
H_{s_{0}}:J_{s}\left( t\right) =\lambda \gamma _{s}^{\prime }\left( t\right)
,\lambda \in \mathbb{R}\right\}
\end{equation*}%
We will prove that $A_{s_{0}}$ is defined locally at $s=s_{0}$ by a
differentiable function $t=t_{s_{0}}\left( s\right) $. Define the function
\begin{equation*}
\begin{array}{crcl}
h_{s_{0}}: & K_{s_{0}}\times H_{s_{0}} & \rightarrow & \mathbb{R} \\
& \left( s,t\right) & \mapsto & \mathbf{g}\left( J_{s}\left( t\right)
,J_{s}\left( t\right) \right)%
\end{array}%
\end{equation*}%
where $\mathbf{g}$ denotes the metric in $M$, and define the set
\begin{equation*}
\widehat{A}_{s_{0}}=\left\{ \left( s,t\right) \in K_{s_{0}}\times
H_{s_{0}}:h_{s_{0}}\left( s,t\right) =0\right\}
\end{equation*}

It is clear that $h_{s_{0}}$ is differentiable and $A_{s_{0}}\subset
\widehat{A}_{s_{0}}$. To prove that $\widehat{A}_{s_{0}}\subset A_{s_{0}}$
consider any $\left( s,t\right) \in \widehat{A}_{s_{0}}$, then
\begin{equation}  \label{equation4.1-1}
\mathbf{g}\left( J_{s}\left( t\right) ,J_{s}\left( t\right) \right) =0
\end{equation}%
but, since the curve $\Gamma $ is celestial, then $\Gamma ^{\prime
}\left(s\right) \in \widehat{\Sigma}$ for every $s\in I$ and $J_{s}\left(
t\right)$ must also verifies
\begin{equation}  \label{equation4.1-2}
\mathbf{g}\left( J_{s}\left( t\right) ,\gamma _{s}^{\prime }\left(
t\right)\right) =0
\end{equation}%
The equations \ref{equation4.1-1} and \ref{equation4.1-2} imply that $%
J_{s}\left(t\right) =\lambda _{s}\gamma _{s}^{\prime }\left( t\right) $ for
some $\lambda _{s}\in \mathbb{R}$, therefore $\left( s,t\right) \in
A_{s_{0}} $.

Calculating $\frac{\partial h_{s_{0}}}{\partial t}\left( s_{0},t_{0}\right) $%
, we obtain
\begin{equation*}
\frac{\partial h_{s_{0}}}{\partial t}\left( s_{0},t_{0}\right) =\left. \frac{%
\partial }{\partial t}\right\vert _{\left( s_{0},t_{0}\right) }\mathbf{g}%
\left( J_{s}\left( t\right) ,J_{s}\left( t\right) \right) =
\end{equation*}%
\begin{equation*}
=2\mathbf{g}\left( \left. \frac{D}{dt}\right\vert _{t_{0}}J_{s_{0}}\left(
t\right) ,J_{s_{0}}\left( t_{0}\right) \right) =2\mathbf{g}\left(
J_{s_{0}}^{\prime }\left( t_{0}\right) ,J_{s_{0}}\left( t_{0}\right) \right)
\end{equation*}%
where $\frac{D}{dt}$ denotes the covariant derivative along $\gamma _{s_{0}}$%
. If we suppose $\frac{\partial h_{s_{0}}}{\partial t}\left(
s_{0},t_{0}\right) =0$ then we have that
\begin{equation*}
0=\mathbf{g}\left( J_{s_{0}}^{\prime }\left( t_{0}\right) ,J_{s_{0}}\left(
t_{0}\right) \right) =\mathbf{g}\left( J_{s_{0}}^{\prime }\left(
t_{0}\right) ,\lambda _{s_{0}}\gamma _{s_{0}}^{\prime }\left( t_{0}\right)
\right)
\end{equation*}%
thereby $J_{s_{0}}^{\prime }\left( t_{0}\right) =\eta _{s_{0}}\gamma
_{s_{0}}^{\prime }\left( t_{0}\right) $ for some $\eta _{s_{0}}\in \mathbb{R}
$. In this case, $J_{s_{0}}\left( t_{0}\right) $ is proportional to $\gamma
_{s_{0}}^{\prime }\left( t_{0}\right) $, then $\Gamma ^{\prime }\left(
s_{0}\right) =0$, but this conflicts with $\Gamma $ is a regular curve. So
we state $\frac{\partial h_{s_{0}}}{\partial t}\left( s_{0},t_{0}\right)
\neq 0$. By implicit function theorem, $h_{s_{0}}\left( s,t\right) =0$
implicitly define a function $t=t_{s_{0}}\left( s\right) $ in an open
neighbourhood $W_{s_{0}}=\left( s_{0}-\epsilon _{s_{0}},s_{0}+\epsilon
_{s_{0}}\right) $ of $s_{0}$. By compactness of the interval $\left[ 0,1%
\right] $, there is a finite covering $\left\{ W_{k}\right\} _{k=1,\ldots
,N} $ of $\left[ 0,1\right] $ with its corresponding functions $\left\{
t_{k}\right\} _{k=1,\ldots ,N}$. For every $s\in W_{k}\cap W_{j}$ then $%
t_{k}\left( s\right) =t_{j}\left( s\right) $ and it is possible to define
the function $t:\left[ 0,1\right] \rightarrow \mathbb{R}$ as
\begin{equation*}
t\left( s\right) =t_{k}\left( s\right) \hspace{1cm}\text{if }s\in W_{k}
\end{equation*}

Now, consider the curve in $M$
\begin{equation*}
\mu \left( s\right) =\mathbf{f}\left( s,t\left( s\right) \right)
\end{equation*}%
The tangent vector $\mu ^{\prime }$ is given by
\begin{equation*}
\mu ^{\prime }\left( s\right) =J_{s}\left( t\left( s\right) \right)
+t^{\prime }\left( s\right) \gamma _{s}^{\prime }\left( t\left( s\right)
\right) =
\end{equation*}%
\begin{equation*}
=\lambda \left( s\right) \gamma _{s}^{\prime }\left( t\left( s\right)
\right) +t^{\prime }\left( s\right) \gamma _{s}^{\prime }\left( t\left(
s\right) \right) =\left[ \lambda \left( s\right) +t^{\prime }\left( s\right) %
\right] \gamma _{s}^{\prime }\left( t\left( s\right) \right)
\end{equation*}%
So we have that $\mu $ is a null curve and $\sigma \left( s\right) =$ $%
\gamma _{s}^{\prime }\left( t\left( s\right) \right) $ is a differentiable
curve which is proportional to $%
\mu ^{\prime }\left( s\right) $ where $\mu $ is regular. If we define $%
\overline{\mathbf{f}}\left( s,\tau \right) =\exp _{\mu \left( s\right)
}\left( \tau \sigma \left( s\right) \right) $, we have that $\overline{%
\mathbf{f}}\left( s,0\right) =\mathbf{f}\left( s,t\left( s\right) \right)
=\mu \left( s\right) $. Then
\begin{equation*}
\overline{\mathbf{f}}\left( s,\tau \right) =\exp _{\mu \left( s\right)
}\left( \tau \sigma \left( s\right) \right) =\exp _{\mu \left( s\right)
}\left( \tau \gamma _{s}^{\prime }\left( t\left( s\right) \right) \right) =%
\overline{\gamma }_{s}\left( \tau \right)
\end{equation*}%
and $\overline{\gamma }_{s}$ is a re-parametrization of $\gamma _{s}$.
Therefore $\Gamma\left( s\right) =\overline{\gamma _{s}}=\gamma _{s}\in \mathcal{N}$.
\end{proof}

The previous proposition describes a celestial curve $\Gamma$ as a pair $%
\left(\mu , \sigma\right)\subset M \times \mathbb{N}^{+}$ where $\mu$ is a null
curve that cannot be geodesic because in this case $\Gamma$ would not be
regular.  Moreover the regularity of $\mu$ is not guaranteed at all, in fact, it is
possible to exhibit examples of celestial curves such that $\mu $ stops for $%
s\in\left[a,b\right]\subset \mathbb{R}$ where $a=b$ is not excluded. While $%
\mu$ remains at $\mu\left(s\right)=p\in M$, the curve $\sigma\left(s\right)$
moves smoothly in $\mathbb{N}^{+}_{p}$. The time-orientation of $\mu$ is not
guaranteed neither, as the next example shows.

\begin{example}
\label{example-mu} Let $\mathbb{M}^3$ be the 3--dimensional Minkowski
space--time with coordinates given by $\left(t,x,y\right)\in \mathbb{R}^3$
and metric $\mathbf{g}=-dt\otimes dt + dx\otimes dx +dy\otimes dy$. Let us
denote its space of  light rays as $\mathcal{N}$. We consider the curve $%
\Gamma:\left[-\varepsilon, \varepsilon\right] \rightarrow \mathcal{N}$
defined by the geodesic variation
\begin{equation*}
\mathbf{f}\left(s,\tau\right)=\gamma_{s}\left(\tau\right)=\left(\tau+\frac{1%
}{2}s^2 , s\sin s+ \left(1+\tau\right)\cos s , -s\cos s+
\left(1+\tau\right)\sin s \right)
\end{equation*}
as $\Gamma\left(s\right)=\gamma_s$. An easy calculation shows that $\Gamma$
is a celestial curve. For this curve, $\mu$ is defined as
\begin{equation*}
\mu\left(s\right)=\mathbf{f}\left(s,\tau\left(s\right)\right)=\mathbf{f}%
\left(s,0\right)=\left( \frac{1}{2}s^2 , s\sin s+ \cos s , -s\cos s+ \sin
s\right)
\end{equation*}
hence,
\begin{equation*}
\mu^{\prime }\left(s\right)=\left( s, s\cos s , s\sin s\right)=s\left( 1,
\cos s , \sin s\right)
\end{equation*}
and $\mu$ is a null curve since
\begin{equation*}
\mathbf{g}\left(\mu^{\prime }\left(s\right) , \mu^{\prime }\left(s\right)
\right)=0
\end{equation*}
but the $s$ factor in $\mu ^{\prime }$ changes the time--orientation of $\mu$%
: if $s<0$ then $\mu$ is past--oriented and if $s>0$ then $\mu$ is
future--oriented. It is trivial to observe that $\mu$ is not a regular curve
when $s=0$.
\end{example}

The previous example motivates the following definitions \ref{definition5.5}-\ref{definition5.7}.

\begin{definition}
\label{definition5.5} With the same notations used in Proposition \ref%
{proposition4.3}, a celestial curve $\Gamma \subset \mathcal{N}$ is called a
\emph{sky curve} if $\Gamma \subset X$ for some sky $X\in \Sigma$. We denote
the set of all sky curves as $\mathfrak{C}_s\left(\mathcal{N}\right)$.
\end{definition}

\begin{definition}
We say that $\left( M,\mathcal{C}\right) $ is \emph{null non-conjugate} if
there are no conjugate points in any null geodesic segment or, equivalently,
if $\widehat{T}X \cap \widehat{T}Y \neq \emptyset$ for two skies $X, Y$ lying on a null geodesic segment, then $X = Y$.
\end{definition}

It is easy to prove that the property of being null non-conjugate does not depend on the chosen auxiliary
metric $\mathbf{g}\in \mathcal{C}$.  Notice that a convex normal neighbourhood
$V$ at any point $x\in M$ is null non-conjugate because it is normal (recall Def. \ref{normal_neigh}) and similarly, a neighbourhood ``small'' enough of any closed spacial surface has this property too.

By convention, we can consider $M \subset \mathfrak{L}\left(M\right)$ since
any point $p\in M$ can be identified with a constant curve.   Moreover, if $M$ is
null non--conjugate, then the
map $\pi_{CL}: \mathfrak{C}\left(\mathcal{N}\right) \rightarrow \mathfrak{L}%
\left(M\right)$ given by $\pi_{CL}\left(\Gamma\right)=\mu$ is well defined
and $\mu$ is characterized by $\Gamma ^{\prime }\left(s\right)\in
\widehat{T}_{\Gamma\left(s\right)}S\left(\mu\left(s\right)\right)$ for every
$s$ \footnote{In the general case, $\Gamma \in \mathfrak{C}\left(\mathcal{N}\right)$ can be
defined by several curves $\mu_i$ with $i=1,2,\ldots $, and so $%
\pi_{CL}\left(\Gamma\right)$ should be interpreted as the family $\lbrace
\mu_i \rbrace$.}.  We call $\{ S(\mu(s))\}$ the Legendrian isotopy of $\Gamma$.

\begin{definition}
\label{definition5.7}  Let $(\mathcal{N}, \Sigma)$ the space of rays and skies of a null non-conjugate strongly causal space--time $M$.
We define the set of \emph{causal celestial curves} as
\begin{equation*}
\mathfrak{C}_c \left(\mathcal{N}\right)= \left\{ \Gamma\in \mathfrak{C}
\left(\mathcal{N}\right) : \mu =  \pi_{CL}\left(\Gamma\right) \in 
\mathfrak{L}_c\left(M\right) \right\}
\end{equation*}
\end{definition}

Definition \ref{definition5.7} of the class of causal celestial curves in $\mathcal{N}$ uses explicitly the space $M$, however because of the
results of Section \ref{section-leg-isot} we can provide a characterization of $%
\mathfrak{C}_{c}\left( \mathcal{N}\right) $ without making any reference to $M$. In
fact, using  Corolary \ref{cor00300} and Propositions \ref{prop00300}, \ref{proposition4.3}, we see that $\mu \in \mathfrak{L}%
_{c}\left( M\right) $ if and only if $\mu $ is a null curve defining a non-positive (or non--negative)
legendrian isotopy and we get the following corollary that could be used as an alternative definition of $\mathfrak{C}_c \left(\mathcal{N}\right)$.

\begin{corollary}\label{causal_legendrian}
A celestial curve $\Gamma \in \mathfrak{C} \left(\mathcal{N}\right)$ is a
past (future) causal celestial curve if and only if $\Gamma$ defines a non-negative (non-positive) legendrian
isotopy of skies.
\end{corollary}

\begin{definition}
Let $M_{1}$ and $M_{2}$ be two strongly causal spaces and let $\mathcal{%
N}_{1}$ and $\mathcal{N}_{2}$ be their corresponding spaces of light rays.
A diffeomorphism $\phi :\mathcal{N}_{1}\rightarrow \mathcal{N}%
_{2}$ will be called a celestial map if it preserves celestial vectors, (i.e. $\phi _{\ast
}\left( \widehat{\Sigma }_{1}\right) \subset \widehat{\Sigma }_{2}$).  \end{definition}

The following Lemma is a direct consequence of the definitions.

\begin{lemma}\label{lemma5.12}
Any celestial map
$\phi :\mathcal{N}_{1}\rightarrow \mathcal{N}_{2}$ preserves celestial curves.
\end{lemma}

\begin{proof}
If $\Gamma :I\rightarrow \mathcal{N}_{1}$ is a celestial curve, then $\Gamma
^{\prime }\left( s\right) \in \widehat{\Sigma }_{1}$ for every $s\in I$. Since $\phi$ is celestial then $\left( \phi \circ \Gamma \right)
^{\prime }\left( s\right) =\phi _{\ast }\left( \Gamma ^{\prime }\left(
s\right) \right) \in \widehat{\Sigma }_{2}$ and hence, $\phi \circ \Gamma
:I\rightarrow \mathcal{N}_{2}$ is a celestial curve. Moreover $\phi$ induces a map
$ \phi :\mathfrak{C}\left( \mathcal{N}_{1}\right) \rightarrow \mathfrak{C}\left( \mathcal{N}_{2}\right)$.
\end{proof}

Finally we have the following definition:

\begin{definition} Let $M_{1}$ and $M_{2}$ be two strongly causal spaces and let $\mathcal{%
N}_{1}$ and $\mathcal{N}_{2}$ be their corresponding spaces of  light rays.
A celestial map $\phi :\mathcal{N}_{1}\rightarrow \mathcal{N}%
_{2}$ will be called a causal celestial map if $\phi $ preserves causal
celestial curves, that is
\begin{equation*}
\phi :\mathfrak{C}_{c}\left( \mathcal{N}_{1}\right) \rightarrow \mathfrak{C}%
_{c}\left( \mathcal{N}_{2}\right)
\end{equation*}%
\end{definition}

\begin{theorem}
\label{lemma5.13} Let $M_{1}$ and $M_{2}$ be two strongly causal spaces,
suppose that $M_{2}$ is null non-conjugate, and let $\left( \mathcal{N}%
_{1},\Sigma _{1}\right) $ and $\left( \mathcal{N}_{2},\Sigma _{2}\right) $
be their corresponding pairs of spaces of  light rays and skies. Let $%
\phi :\mathcal{N}_{1}\rightarrow \mathcal{N}_{2}$ be a celestial map. Then
the following conditions are equivalent:

\begin{enumerate}
\item \label{TR-i} $\phi $ is a  causal celestial map, that is $\phi \circ
\Gamma _{1}\in \mathfrak{C}_{c}\left( \mathcal{N}_{2}\right) $, for all
$\Gamma _{1}\in \mathfrak{C}_{c}\left( \mathcal{N}_{1}\right) $

\item \label{TR-ii} $\phi $ is a celestial sky map, that is $\phi \circ
\Gamma _{1}\in \mathfrak{C}_{s}\left( \mathcal{N}_{2}\right)$, for all $\Gamma
_{1}\in \mathfrak{C}_{s}\left( \mathcal{N}_{1}\right) $.

\item \label{TR-iii} There exists a conformal immersion $\Phi
:M_{1}\rightarrow M_{2}$ such that $\phi \left( \gamma \right) =\Phi \circ
\gamma $ for every $\gamma \in \mathcal{N}_{1}$.
\end{enumerate}
\end{theorem}

\begin{proof}
(\ref{TR-i}) $\Rightarrow $ (\ref{TR-ii}) and (\ref{TR-iii}) $\Rightarrow $ (%
\ref{TR-i}) are trivial.

(\ref{TR-ii}) $\Rightarrow $ (\ref{TR-iii}) Consider $X_{1}\in \Sigma _{1}$ and a closed sky curve $\Gamma
_{1}\in \mathfrak{C}_{s}\left( \mathcal{N}_{1}\right) $ such that $\Gamma
_{1}:\left[ 0,1\right] \rightarrow X_{1}\subset \mathcal{N}_{1}$. Since $\phi $ is a diffeomorphism and by lemma \ref%
{lemma5.12}, then $\Gamma _{2}=\phi \circ \Gamma _{1}$ is a closed celestial
curve. Let $\mu _{2}$ and $\sigma _{2}$ be the curves defining $\Gamma _{2}$%
, according to proposition \ref{proposition4.3}. Then, the endpoints verify
\begin{equation*}
\mu _{2}\left( 0\right) ,\mu _{2}\left( 1\right) \in \Gamma _{2}\left(
0\right) =\Gamma _{2}\left( 1\right) =\gamma _{2}\in \mathcal{N}_{2}
\end{equation*}%
By the hypothesis we have that $\Gamma _{2}\in \mathfrak{C}_{c}\left(
\mathcal{N}_{2}\right) $ and therefore $\mu _{2}\in \mathfrak{L}_{c}\left(
M\right) $ . We will show that $\mu _{2}$ is a constant, and therefore that $%
\Gamma _{2}$ is a sky curve. Suppose that $\mu _{2}$ is no constant, then
we can construct a future null curve $\overline{\mu }_{2}$ such that $%
\mathrm{Im}\left( \overline{\mu }_{2}\right) =\mathrm{Im}\left( \mu
_{2}\right) $ and $\mu _{2}\left( 0\right) ,\mu _{2}\left( 1\right)
\in \gamma _{2}\cap \overline{\mu }_{2}$. Since $M_{2}$ is strongly
causal, then $\mu _{2}\left( 0\right) \neq \mu _{2}\left( 1\right) $ and by
\cite[Prop. 10.51]{On83}, $\mu _{2}\left( 0\right)$ and $\mu _{2}\left( 1\right)$
are timelikely related and there exists a conjugate point of $\mu_2\left(0\right)$ in $\gamma _{2}$
before $\mu_2\left(1\right)$ contradicting that $M_2$ is conformal non-conjugate.
Therefore $\mu_2$ must be constant. This shows that $\phi$ preserves sky curves and hence
also skies. Then Thm. \ref{theorem4.6} gives us the desired result.
\end{proof}

The following example illustrates that the existence of a contactomorphism preserving celestial vectors between the spaces of  light rays of two space--times is not sufficient to induce a conformal diffeomorphism (on its image) between them, showing that condition (1) in Thm. \ref{lemma5.13} cannot be weakened.

\begin{example}
Let $M=\mathbb{M}^3$ be the 3--dimensional Minkowski space--time with coordinates given by $\left(t,x,y\right)\in \mathbb{R}^3$ and let $\mathcal{N}$ be its space of  light rays.
The hypersurface $C \equiv \left\{t=0\right\}$ is a Cauchy surface, then $\left(x,y,\theta\right)\in \mathbb{R}^2
\times \mathbb{S}^1$ are coordinates in $\mathcal{N}$ for any null geodesic $\gamma\left(s\right)=\left(s, x+s\cos \theta,y+s\sin \theta \right)$.
Then $\left\{ \left(\frac{\partial}{\partial x}\right)_{\gamma}, \left(\frac{\partial}{\partial y}\right)_{\gamma}, \left(\frac{\partial}{\partial \theta}\right)_{\gamma} \right\}$ is a basis of $T_{\gamma}\mathcal{N}$.
The contact hyperplane $\mathcal{H}_{\gamma}$ is generated by the tangent spaces of two different skies containing $\gamma$, therefore
\[
\mathcal{H}_{\gamma}=\mathrm{span}\left\{ \left(\frac{\partial}{\partial \theta}\right)_{\gamma}, \sin \theta \left(\frac{\partial}{\partial x}\right)_{\gamma} - \cos \theta \left(\frac{\partial}{\partial y}\right)_{\gamma} \right\}
\]
and a contact form $\alpha$ can be written as
\[
\alpha = \cos \theta dx + \sin \theta dy
\]
For this $\gamma$, we have that $T_{\gamma}S\left(\gamma\left(s\right)\right)=\mathrm{span}\left\{ s \left(\sin \theta \left(\frac{\partial}{\partial x}\right)_{\gamma} - \cos \theta \left(\frac{\partial}{\partial y}\right)_{\gamma}\right) + \left(\frac{\partial}{\partial \theta}\right)_{\gamma} \right\}$ with $s\in \mathbb{R}$ and hence the celestial vectors at $\gamma$ are given by $\widetilde{\gamma}=\bigcup_{s\in\mathbb{R}}T_{\gamma}S\left(\gamma\left(s\right)\right)$.
It can be easily observed that the whole $\mathcal{H}_{\gamma}$ is covered by $\widetilde{\gamma}$ except the subspace $\mathrm{span}\left\{ \sin \theta \left(\frac{\partial}{\partial x}\right)_{\gamma} - \cos \theta \left(\frac{\partial}{\partial y}\right)_{\gamma} \right\} $.

We can restrict this space to $M_0=\left\{\left(t,x,y\right)\in \mathbb{M}^3 : t<0 \right\}$ denoting $\mathcal{N}_0$ its corresponding space of  light rays.
By global hyperbolicity of $M$ and $M_0$, every null geodesic $\gamma_0 \in \mathcal{N}_0$ can be written as $\gamma_0 = \gamma \cap M_0$ for a unique null geodesic $\gamma \in \mathcal{N}$, then we can define the restriction map
\begin{equation*}
\begin{tabular}{rccl}
$\rho :$ & $\mathcal{N}$ & $\longrightarrow$ & $\mathcal{N}_0$ \\
 & $\gamma$   & $\longmapsto$  & $\gamma_0 = \gamma \cap M_0$
\end{tabular}
\end{equation*}
and the extension map
\begin{equation*}
\begin{tabular}{rccl}
$\varepsilon :$ & $\mathcal{N}_0$ & $\longrightarrow$ & $\mathcal{N}$ \\
         & $\gamma_0$          & $\longmapsto$     & $\gamma$
\end{tabular}
\end{equation*}
Both $\rho$ and $\varepsilon$ are contactomorphisms and they verify $\varepsilon = \rho^{-1}$ and hence we have that $\mathcal{N} \simeq \mathcal{N}_0$.

Now, let us consider $M_{\epsilon}=\left\{\left(t,x,y\right)\in \mathbb{R}^3 : t<\epsilon \right\}$ for $\epsilon >0$, equipped with the metric
\[
\mathbf{g}_{\epsilon} = -\left(1+f\left(t\right)\right)dt\otimes dt + 2f\left(t\right)dt\otimes dx + \left(1-f\left(t\right)\right)dx\otimes dx+dy\otimes dy
\]
where $f$ is a smooth function verifying $f\left(t\right)=0$ for every $t\leq 0$.
We can see $\mathbf{g}_{\epsilon}$ as a small perturbation of the metric $\mathbf{g}$ of $M$ for $0<t<\epsilon$.
Trivially, we observe that $M$ and $M_{\epsilon}$ are two space--times extending $M_0$.
By \cite{NM11}, the value of $\epsilon$ can be chosen small enough such that $M_{\epsilon}$ remains globally hyperbolic, then we can consider $\mathcal{N}_{\epsilon} \simeq \mathcal{N}$ and therefore $\mathcal{H}_{\gamma} \simeq \mathcal{H}_{\gamma_{0}} \simeq \mathcal{H}_{\gamma_{\epsilon}}$ for $\gamma_0 =\gamma \cap M_0$ and $\gamma_{\epsilon} =\gamma \cap M_{\epsilon}$. This extension is independent from the coordinates $x$ and $y$. Denoting by $\widetilde{\gamma_{\epsilon}}$, $\widetilde{\gamma_0}$ the celestial vectors at the corresponding curve, and working at $\mathcal{N}$ with certain notation abuse we have that $\widetilde{\gamma_0}=\bigcup_{s\in\left(-\infty,0\right)}T_{\gamma}S\left(\gamma\left(s\right)\right)\subset \widetilde{\gamma}\cap \widetilde{\gamma_{\epsilon}}$ then the value $\epsilon$ also can be selected small enough such that $\widetilde{\gamma_{\epsilon}}\subset \widetilde{\gamma}$ and therefore the contactomorphism $\Phi:\mathcal{N}_{\epsilon}\rightarrow\mathcal{N}$ preserves celestial vectors.
In spite of the existence of $\Phi$ preserving celestial vectors, the space--times $M$ and $M_{\epsilon}$ can not be conformally equivalent.
Observe that 3--dimensional Minkowski space--time $M$ is flat. Denoting as $R_{ij}$, $R$ and $g^{\epsilon}_{ij}$ the Ricci curvature, the scalar curvature and the metric in $M_{\epsilon}$ respectively, then the components of the Cotton tensor $\mathbf{C}_{\epsilon}$ in $M_{\epsilon}$ are given by $C_{ijk}=\nabla _{k}R_{ij}-\nabla _{j}R_{ik}+\frac{1}{4}\left( \nabla_{j}Rg^{\epsilon}_{ik}-\nabla _{k}Rg^{\epsilon}_{ij}\right)$. It is widely known that one 3--dimensional manifold is locally conformally flat if its Cotton tensor vanishes.
A straightforward calculation shows that $\mathbf{C}_{\epsilon}\neq 0$, then $M_{\epsilon}$ is not conformally flat and therefore it can not be conformal to $M$.
\end{example}


\newpage

\begin{thebibliography}{Ch08}

\bibitem[Ab88]{Ab88} R. Abraham, J. Marsden, T. Ratiu, \emph{Manifolds,
tensor analysis, and applications}. Springer-Verlag, 1988.

\bibitem[BE96]{BE96} J.K. Beem, P.E.Ehrlich, K.L.Easley. \emph{Global
Lorentzian Geometry}. Marcel Dekker, New York, 1996.

\bibitem[DC92]{DC92} M.P. Do Carmo. \emph{Riemannian Geometry}. Birkh\"{a}%
user, Boston, 1992.

\bibitem[Ch08]{Ch08} V. Chernov, Yu. Rudyak, Linking and causality in
globally hyperbolic space-times, Comm. Math. Phys. 279 (2008), 309-354.

\bibitem[Ch10]{Ch10} V. Chernov, S. Nemirovski. \emph{Legendrian Links,
Causality, and the Low Conjecture}. Geom. Funct. Analysis, \textbf{19} (5)
1320-1333 (2010).

\bibitem[Ha94]{Ha94} S.G. Harris. \emph{The method of timelike 2--surfaces}.
Contempt. Math., \textbf{170}, 125--34 (1994).

\bibitem[Ha01]{Ha01} S.G. Harris, R.J. Low. \emph{Causal monotonicity,
omniscient foliations and the shape of space}. Class. Quantum Grav., \textbf{%
18} 27--43 (2001).

\bibitem[Ha64]{Ha64} P. Hartman. \emph{Ordinary Differentiable Equations}.
John Wiley \& Sons, New York, 1964.

\bibitem[HE73]{HE73} S.W. Hawking, G.F.R. Ellis. \emph{The large scale
structure of space-time}. Cambridge University Press, Cambridge, 1973.

\bibitem[Lo88]{Lo88} R. J. Low, Causal relations and spaces of null
geodesics, PhD Thesis, Oxford University (1988).

\bibitem[Lo89]{Lo89} R. J. Low, The geometry of the space of null geodesics,
J. Math. Phys. 30(4) (1989), 809-811.

\bibitem[Lo90]{Lo90} R. J. Low, Twistor linking and causal relations,
Classical Quantum Gravity 7 (1990), 177-187.

\bibitem[Lo93]{Lo93} R. J. Low, Celestial spheres, light cones, and cuts, J.
Math. Phys. 34, 315 (1993).

\bibitem[Lo94]{Lo94} R. J. Low, Twistor linking and causal relations in
exterior Schwarzschild space, Classical Quantum Gravity 11 (1994), 453-456.

\bibitem[Lo98]{Lo98} R. J. Low, Stable singularities of wave-fronts in
general relativity, J. Math. Phys. 39 (1998), 3332-3335.

\bibitem[Lo00]{Lo00} R. J. Low, The space of null geodesics, Proceedings of
the Third World Congress of Nonlinear Analysts, Part 5 (Catania, 2000).
Nonlinear Anal. 47 (2001), 3005-3017.

\bibitem[Lo06]{Lo06} R. J. Low, The space of null geodesics (and a new
causal boundary), Lecture Notes in Physics 692, Springer, Berlin Heidelberg
New York, 2006, pp. 35-50.

\bibitem[Mi08]{Mi08} E. Minguzzi, M. S\'{a}nchez, \emph{The causal hierarchy
of spacetimes}, Zurich: Eur. Math. Soc. Publ. House, vol. H. Baum, D. Alekseevsky (eds.), Recent
developments in pseudo-Riemannian geometry of ESI Lect. Math. Phys.,
pages 299–358 (2008). arXiv:gr-qc/0609119.

\bibitem[On83]{On83} B. O'Neill, \emph{Semi-Riemannian geometry with
applications to Relativity}. Academic Press. New York, 1983.

\bibitem[Po12]{Po12}  R. Pourkhandani, Y. Bahrampour, \emph{The space of causal curves
and separation axioms}.  Class. Quantum Grav. {\bf 29}, 015014 (2012).


\end{thebibliography}
\end{document}